\documentclass[aos]{imsart}

\usepackage[english]{babel}

\usepackage[boxruled]{algorithm2e}

\usepackage{amsfonts}
\newcommand{\E}{\mathbb E}
\newcommand{\N}{\mathbb N}
\newcommand{\R}{\mathbb R}
\renewcommand{\P}{\mathbb P}

\usepackage{bbm}

\usepackage{amsmath}

\usepackage{amsthm}
\newtheorem{theorem}{Theorem}[section]

\newtheorem{proposition}[theorem]{Proposition}
\newtheorem{assumption}[theorem]{Assumption}

\theoremstyle{definition}

\theoremstyle{remark}
\newtheorem{example}[theorem]{Example}
\newtheorem{remark}[theorem]{Remark}

\usepackage{graphicx}
\usepackage{subcaption}

\usepackage{hyperref}

\usepackage{xcolor}

\RequirePackage{natbib}

\newcommand{\argmin}{\operatorname{argmin}}

\begin{document}
\begin{frontmatter}

\title{The Zig-Zag Process and Super-Efficient Sampling for Bayesian
Analysis of Big Data}
\runtitle{Zig-Zag Sampling}

\begin{aug}
  \author{\fnms{Joris}  \snm{Bierkens}\corref{}\thanksref{t1}\ead[label=e1]{joris.bierkens@tudelft.nl}},
  \author{\fnms{Paul} \snm{Fearnhead}\thanksref{t1}\ead[label=e2]{p.fearnhead@lancs.ac.uk}}
  \and
  \author{\fnms{Gareth}  \snm{Roberts}\thanksref{t1}\ead[label=e3]{gareth.o.roberts@warwick.ac.uk}}%

  \thankstext{t1}{The authors acknowledge the EPSRC for support under grants EP/D002060/1 (CRiSM) and EP/K014463/1 (iLike).}

  \runauthor{J. Bierkens, P. Fearnhead and G. O. Roberts}

  \affiliation{Delft University of Technology, Lancaster University and University of Warwick}

  \address{Delft Institute of Applied Mathematics, \\
  Van Mourik Broekmanweg 6, \\
  2628 XE Delft, \\
  Netherlands, \\ \printead{e1}}
  
\address{Department of Mathematics and Statistics,\\ 
	    Fylde College, \\
Lancaster University, \\
Lancaster, LA1 4YF, \\
United Kingdom, \\ \printead{e2}
}
  \address{Department of Statistics,\\ 
  University of Warwick, \\
  Coventry CV4 7AL, \\
  United Kingdom, \\ \printead{e3}
      }

\end{aug}

%
%
%
%
%

\begin{abstract} ~ 
Standard MCMC methods can scale poorly to big data settings due to the need to evaluate the likelihood at each iteration. There have been a number of approximate MCMC algorithms that use sub-sampling ideas to reduce this computational burden, but with the drawback that these algorithms no longer target the true posterior distribution. We introduce a new family of Monte Carlo methods based upon a multi-dimensional version of the Zig-Zag process of \cite{BierkensRoberts2015}, a continuous time piecewise deterministic Markov process. While traditional MCMC methods are reversible by construction (a property which is known to inhibit rapid convergence) the Zig-Zag process offers a flexible non-reversible alternative which we observe to often have favourable convergence properties. We show how the Zig-Zag process can be simulated without discretisation error, and give conditions for the process to be ergodic. Most importantly, we introduce a sub-sampling version of the Zig-Zag process that is an example of an {\em exact approximate scheme}, i.e. the resulting approximate process still has the posterior as its stationary distribution. Furthermore, if we use a control-variate idea to reduce the variance of our unbiased estimator, then the Zig-Zag process can be super-efficient: after an initial pre-processing step, essentially independent samples from the posterior distribution are obtained at a computational cost which does not depend on the size of the data. 
\end{abstract}

\begin{keyword}[class=MSC]
\kwd[Primary ]{65C60}; \ 
\kwd[secondary ]{65C05} 
\kwd{62F15} 
\kwd{60J25}  
\end{keyword}

\begin{keyword}
\kwd{MCMC}
\kwd{non-reversible Markov process}
\kwd{piecewise deterministic Markov process}
\kwd{Stochastic Gradient Langevin Dynamics}
\kwd{sub-sampling}
\kwd{exact sampling}
\end{keyword}

\end{frontmatter}





\section{Introduction}

The importance of Markov chain Monte Carlo techniques in Bayesian inference shows no signs of diminishing. However, all
commonly used methods are variants on the Metropolis-Hastings (MH) 
algorithm  \cite[]{Metropolis1953,Hastings1970} and rely on innovations which date back over 60 years.
All MH algorithms simulate realisations from a discrete reversible ergodic Markov chain with invariant distribution $\pi$
which is (or is closely related to) the {\em target} distribution, i.e. the posterior distribution in a Bayesian context. The MH algorithm gives a beautifully simple though flexible recipe for constructing such Markov chains, requiring only local information about $\pi$ (typically pointwise evaluations of $\pi$ and, perhaps, its derivative at the current and proposed new locations) to complete each iteration.

However new complex modelling and data paradigms are seriously challenging these established methodologies. Firstly, the restriction of traditional
MCMC to reversible Markov chains is a serious limitation. It is now well-understood both theoretically 
\cite[]{Hwang1993,Chen2013,ReyBelletSpiliopoulos2015,Bierkens2015,DuncanLelievrePavliotis2015} and heuristically \cite[]{neal1998suppressing} 
that non-reversible chains offer potentially massive advantages over reversible counterparts. The need to escape reversibility, and create momentum to aid mixing throughout the state space is certainly well-known, 
and motivates a number of 
modern MCMC methods, including the popular Hamiltonian MCMC \cite[]{Duane1987}. 

A second major obstacle to the application of MCMC for Bayesian inference is the need to process potentially massive data-sets. 
Since MH algorithms in their pure form require a likelihood evaluation -- and thus processing the full data-set -- at every iteration, 
it can be impractical to carry out large numbers of MH iterations. 
This has led to a range of alternatives that use sub-samples of the data at each iteration 
\cite[]{WellingTeh2011,maclaurin2014firefly,Ma:2015,Quiroz:2015}, or that partition 
the data into shards, run MCMC on each shard, and then attempt to combine the information from these different MCMC runs \cite[]{Neiswanger:2013,Scott:2013,Wang:2013,li2017simple}. However
most of these methods introduce some form of approximation error, so that the final sample will be drawn from some approximation to the posterior, and the quality of the approximation can be impossible to evaluate. 
As an exception the Firefly algorithm \cite[]{maclaurin2014firefly} samples from the exact distribution of interest (but see the comment below).

This paper introduces the multi-dimensional Zig-Zag sampling algorithm (ZZ) and its variants. 
These methods overcome the restrictions of the lifted Markov chain approach of \cite{TuritsynChertkovVucelja2011} as they do not depend on the introduction of momentum generating quantities.
They are also amenable to the use of sub-sampling ideas. The dynamics of the Zig-Zag process depends on the target distribution through the gradient of the logarithm of the target. For Bayesian applications
this is a sum, and is easy to estimate unbiasedly using sub-sampling. 
Moreover, Zig-Zag with Sub-Sampling (ZZ-SS) retains the exactness of the required invariant distribution.
Furthermore, if we also use control variate ideas to reduce the variance of our sub-sampling estimator of the gradient, the resulting Zig-Zag with Control Variates (ZZ-CV) algorithm
has remarkable {\em super-efficient} scaling properties for large data sets.

We will call an algorithm \emph{super-efficient} if it is able to generate independent samples from the target distribution at a higher efficiency than if we would draw independently from the target distribution at 
the cost of evaluating all data. The only situation we are aware of where we can implement super-efficient sampling is with simple conjugate models, where the likelihood function has a low-dimensional
summary statistic which can be evaluated at cost $O(n)$, where $n$ is the number of observations, after which we can obtain independent samples from the posterior distribution at a cost of $O(1)$ by using the functional form of the 
posterior distribution.
The ZZ-CV can replicate this computational efficiency: after a pre-computation of $O(n)$, we are able to obtain independent samples at a cost of $O(1)$. In this sense it contrasts with the Firefly algorithm \cite[]{maclaurin2014firefly} which has an ESS per datum which decreases approximately as $1/n$ where $n$ is the size of the data, so that the gains of this algorithm do not increase with $n$; see \cite[Section 4.6]{Bouchard-Cote2015}.

This breakthrough is based upon the Zig-Zag process, a continuous time piecewise deterministic Markov process (PDMP).
Given a $d$-dimensional differentiable target density $\pi$, Zig-Zag is a continuous-time non-reversible stochastic process with continuous, piecewise linear trajectories on $\R^d$. It moves with constant velocity, $\Theta \in \{-1, 1\}^d$, until 
one of the velocity components switches sign. The event time and choice of which direction to reverse is controlled by a collection of state-dependent
switching rates, $(\lambda_i)_{i=1}^d$ which in turn are constrained via an identity \eqref{eq:relation_lambda} which ensures that $\pi$ is a stationary distribution for the process. The process intrinsically is 
constructed in continuous-time, and it can be easily simulated using standard Poisson thinning arguments as we shall see in Section \ref{sec:sim}.

The use of PDMPs such as the Zig-Zag processes is an exciting and mostly unexplored area in MCMC. The first occurrence of a PDMP for sampling purposes is in the computational physics literature 
\cite[]{PetersDeWith2012}, which in one dimension coincides with the Zig-Zag process. In \cite{Bouchard-Cote2015} this method is given the name \emph{Bouncy Particle Sampler}. In multiple dimensions the Zig-Zag process and Bouncy Particle Sampler (BPS) are different processes: both are PDMPs which move along straight line segments, but the Zig-Zag process changes direction in only a single component at each switch, whereas the Bouncy Particle Sampler reflects the full direction vector in the level curves of the density function. As we will see in Section~\ref{sec:ergodicity}, this difference has a beneficial effect on the ergodic properties of the Zig-Zag process. The one-dimensional Zig-Zag process is analysed in detail in e.g. \cite{Fontbona2012, Monmarche2014b, Fontbona2015, BierkensRoberts2015}. 

Since the first version of this paper was conceived already several other related theoretical and methodological papers have appeared. In particular we mention here results on exponential ergodicity of the BPS \cite[]{Deligiannidis2017} and ergodicity of the multi-dimensional Zig-Zag process \cite[]{Bierkens2017}. The Zig-Zag process has the advantage that it is ergodic under very mild conditions, which in particular means that we are not required to choose a refreshment rate. At the same time, the BPS seems more `natural', in that it tries to minimise the bounce rate and the change in direction at bounces, and it may be more efficient for this reason. However it is a challenge to make a direct comparison in efficiency of the two methods since the efficiency depends both on the computational effort per unit of continuous time of the respective algorithms, as well as the mixing time of the underlying processes. Therefore we expect analysing the relative efficiency of PDMP based  algorithms to be an important area of continued research for years to come.

A continuous-time sequential Monte Carlo algorithm for scalable Bayesian inference with big data, the SCALE algorithm, is given in \cite{Pollock2016}. 
Advantages that Zig-Zag has over SCALE is that it avoids the issue of controlling the stability of importance weights, and it is simpler to implement. Whereas the SCALE algorithm is well-adapted for the use of 
parallel architecture computing, and has particularly simple scaling properties for big data. 


\subsection{Notation}

For a topological space $X$ let $\mathcal B(X)$ denote the Borel $\sigma$-algebra.
We write $\R_+ := [0, \infty)$. 
If $h : \R^d \rightarrow \R$ is differentiable then $\partial_i h$ denotes the function $\xi \mapsto \frac{\partial h(\xi)}{\partial \xi_i}$. 
We equip $E := \R^d \times \{-1,+1\}^d$ with the product topology of the Euclidean topology on $\R^d$ and the discrete topology on $\{-1,+1\}^d$. Elements in $E$ will often be denoted by $(\xi, \theta)$ with $\xi \in \R^d$ and $\theta \in \{-1,+1\}^d$.  For $g : E \rightarrow \R$ differentiable in its first argument we will use  $\partial_{i} g$ to denote the function $(\xi, \theta) \mapsto \frac{\partial g(\xi,\theta)}{\partial \xi_i}$, $i=1,\dots, d$.

\section{The Zig-Zag process}
\label{sec:ZZ}


The Zig-Zag process is a continuous time Markov process whose trajectories lie in the space $E = \R^d \times \{-1,+1\}^d$ and will be denoted by $(\Xi(t), \Theta(t))_{t \geq 0}$. They can be described as follows: at random times a single component of $\Theta(t)$ flips. In between these switches, $\Xi(t)$ is linear with $\frac{d}{dt} \Xi(t) = \Theta(t)$. The rates at which the flips in $\Theta(t)$ occur are time inhomogeneous: the $i$-th component of $\Theta$ switches at rate $\lambda_i(\Xi(t), \Theta(t))$, where $\lambda_i : E \rightarrow \R_+$ for $i =1, \dots, d$ are continuous functions.

\subsection{Construction}
\label{sec:construction}

For a given $(\xi,\theta) \in E$, we may construct a trajectory of $(\Xi,\Theta)$ of the Zig-Zag process with initial condition $(\xi, \theta)$ as follows.

\begin{itemize}
 \item Let $(T^0, \Xi^0, \Theta^0) := (0, \xi,\theta)$.
 \item For $k = 1, 2, \dots$ 
 \begin{itemize}
 \item Let $\xi^k(t) := \Xi^{k-1} + \Theta^{k-1} t$, $t \geq 0$
 \item For $i = 1, \dots, d$, let $\tau^k_i$ be distributed according to
 \[ \P(\tau^k_i \geq t) = \exp \left( - \int_0^t \lambda_i(\xi^k(s), \Theta^{k-1}) \ d s \right).\]
 \item Let $i_0 := \argmin_{i \in \{1, \dots, d \}} \tau^k_i$ and let $T^k := T^{k-1}+\tau^k_{i_0}$.
 \item Let $\Xi^k :=\xi^k(T^k)$.
 \item Let 
 \[ \Theta^k(i) := \left\{ \begin{array}{ll} \Theta^{k-1}(i) \quad & \mbox{if $i \neq i_0$}, \\
                            - \Theta^{k-1}(i) \quad & \mbox{if $i = i_0$}
                           \end{array} \right.\]
\end{itemize}
\end{itemize}
This procedure defines a sequence of \emph{skeleton points} $(T^k, \Xi^k, \Theta^k)_{k = 0}^{\infty}$ in $\R_+ \times E$, which correspond to the time and position at which the direction of the process changes.
The trajectory $\xi^k(t)$ represents the position of the process at time $T^{k-1}+t$ 
until time $T^k$, for $0\le t\le T^k-T^{k-1}$.
The time until the next skeleton event is characterized as the smallest time 
of a set of events in $d$ simultaneous point processes, where each point process corresponds to  switching of a different 
component of the velocity. For the $i$-th of these processes, events occur at rate $\lambda_i(\xi^k(s), \Theta^{k-1})$, and $\tau_i^k$ is defined to be the time to the first event for the $i$-th component. The component for which the earliest event occurs is  $i_0$. This defines $\tau_{i_0}^k$, the time between the $(k-1)$th and $k$th skeleton point, and the component, $i_0$, of the velocity that switches.

The piecewise deterministic trajectories $(\Xi(t), \Theta(t))$ are now obtained as 
\[ (\Xi(t),\Theta(t)) := (\Xi^k + \Theta^k (t- T^k), \Theta^k) \quad \mbox{for $t \in [T^k, T^{k+1})$}, \quad k = 0, 1, 2, \dots. \]

Since the switching rates are continuous and hence bounded on compact sets, and $\Xi$ will travel a finite distance within any finite time interval, within any bounded time interval there will be finitely many switches almost surely.

The above procedure provides a mathematical construction of a Markov process 
as well as the outline of an algorithm which simulates this process. The only step in this procedure which presents a computational challenge is the simulation of the random times $(T_i^k)$ and a significant part of this paper will consider obtaining these in a numerically efficient way.

Figure~\ref{fig:basic-examples} displays trajectories of the Zig-Zag process for several examples of invariant distributions. The name of the process is derived by the 
{\em zig-zag} nature of paths that the process produces.
Figure~\ref{fig:basic-examples} shows an important difference in the output of the Zig-Zag process, as compared to a discrete-time MCMC algorithm: the output of is a continuous-time sample path.  
The bottom row of plots in Figure~\ref{fig:basic-examples}  also gives a comparison to a reversible MCMC algorithm, Metropolis Adjusted Langevin \cite[MALA][]{RobertsTweedie1996}, and demonstrates an advantage of a non-reversible sampler:
it can cope better with a heavy tailed target. This is most easily seen if we start the process out in the tail, as in the figure.
The velocity component of the Zig-Zag process enables it to quickly return to the mode of the distribution, whereas the reversible algorithm behaves like a random walk in the tails, and takes much longer
to return to the mode.



\begin{figure}
\begin{subfigure}[b]{0.45 \textwidth}
 \includegraphics[width=\textwidth]{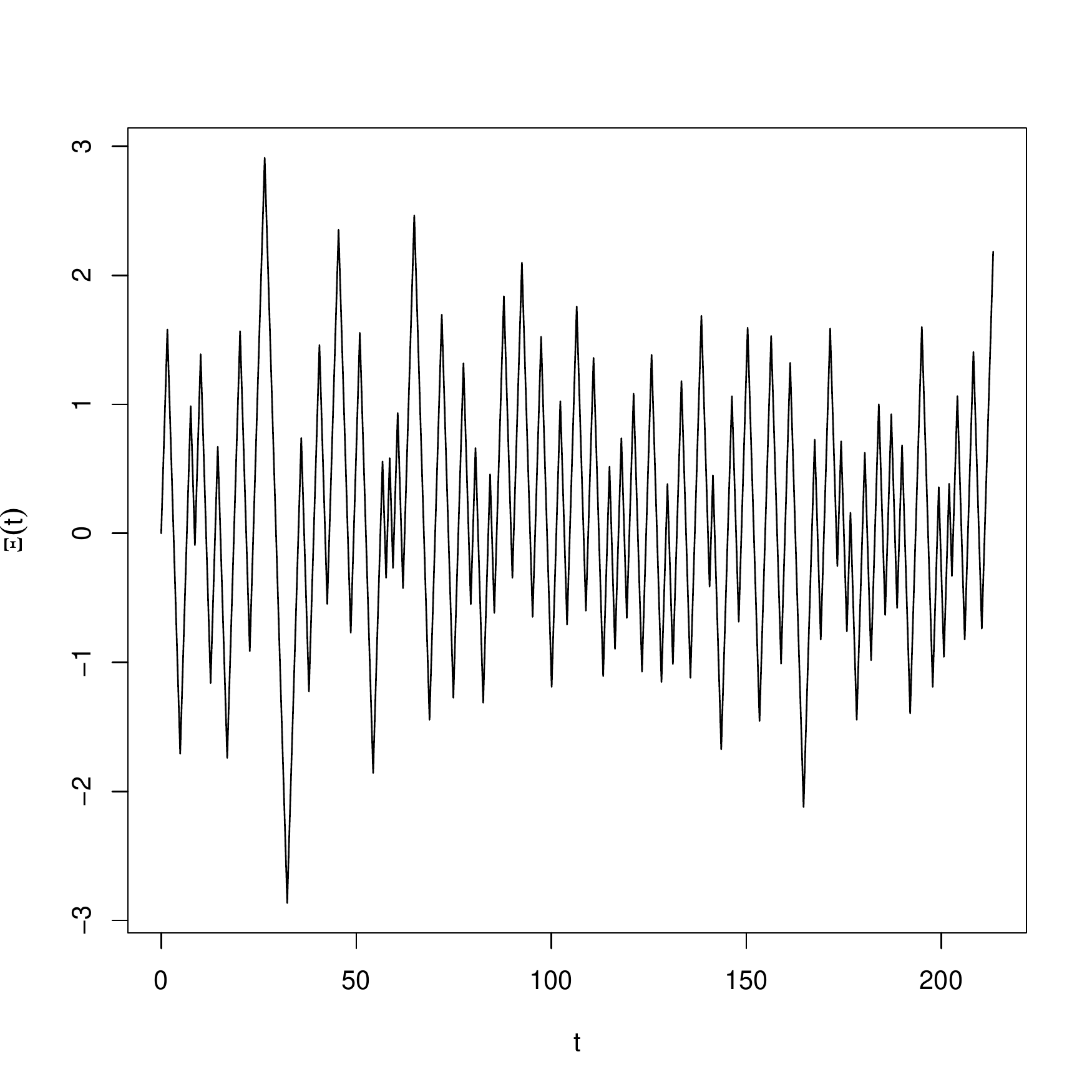}
 \caption{1D Gaussian}
\end{subfigure}
\begin{subfigure}[b]{0.45 \textwidth}
 \includegraphics[width=\textwidth]{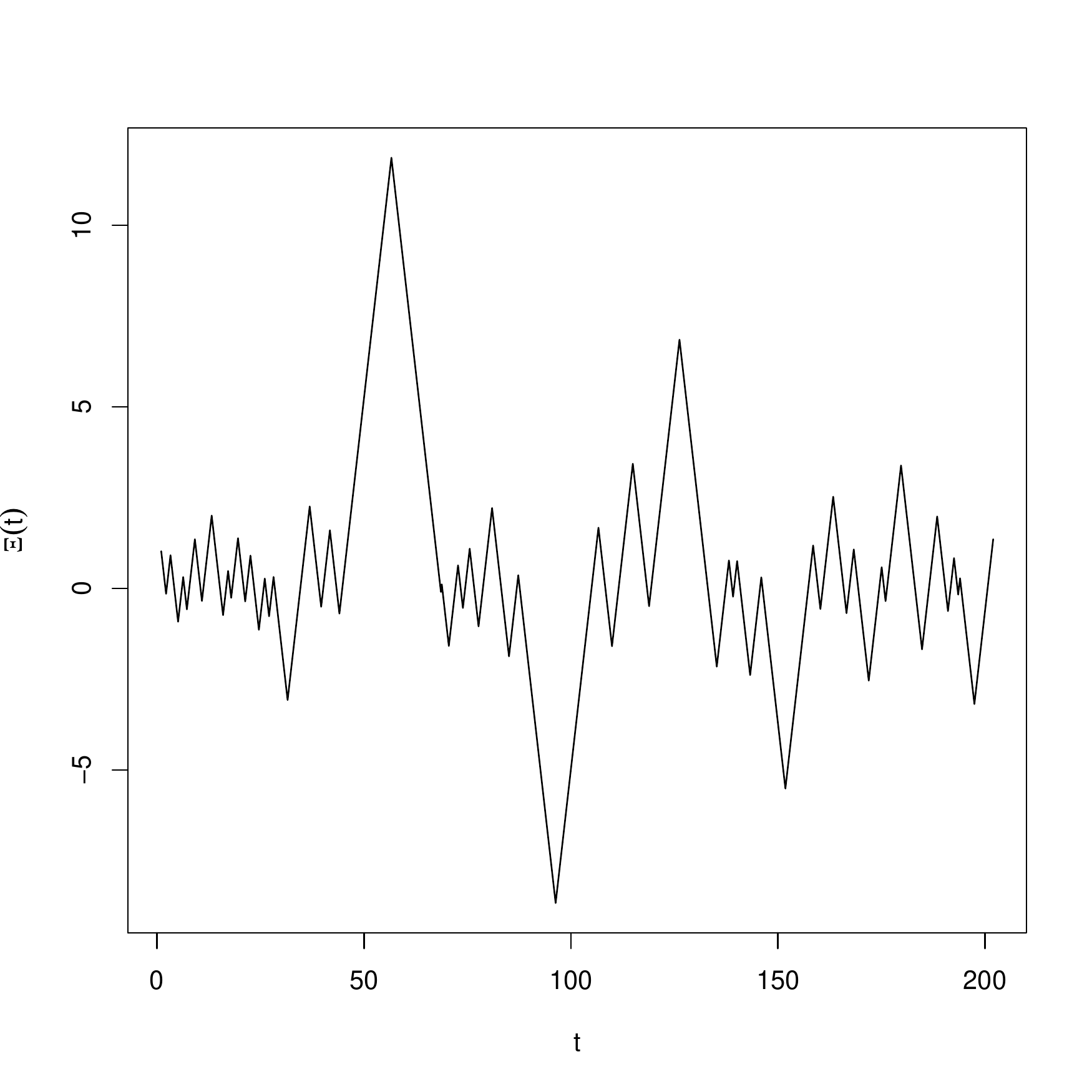}
 \caption{1D Cauchy}
\end{subfigure}

\begin{subfigure}[b]{0.45 \textwidth}
 \includegraphics[width=\textwidth]{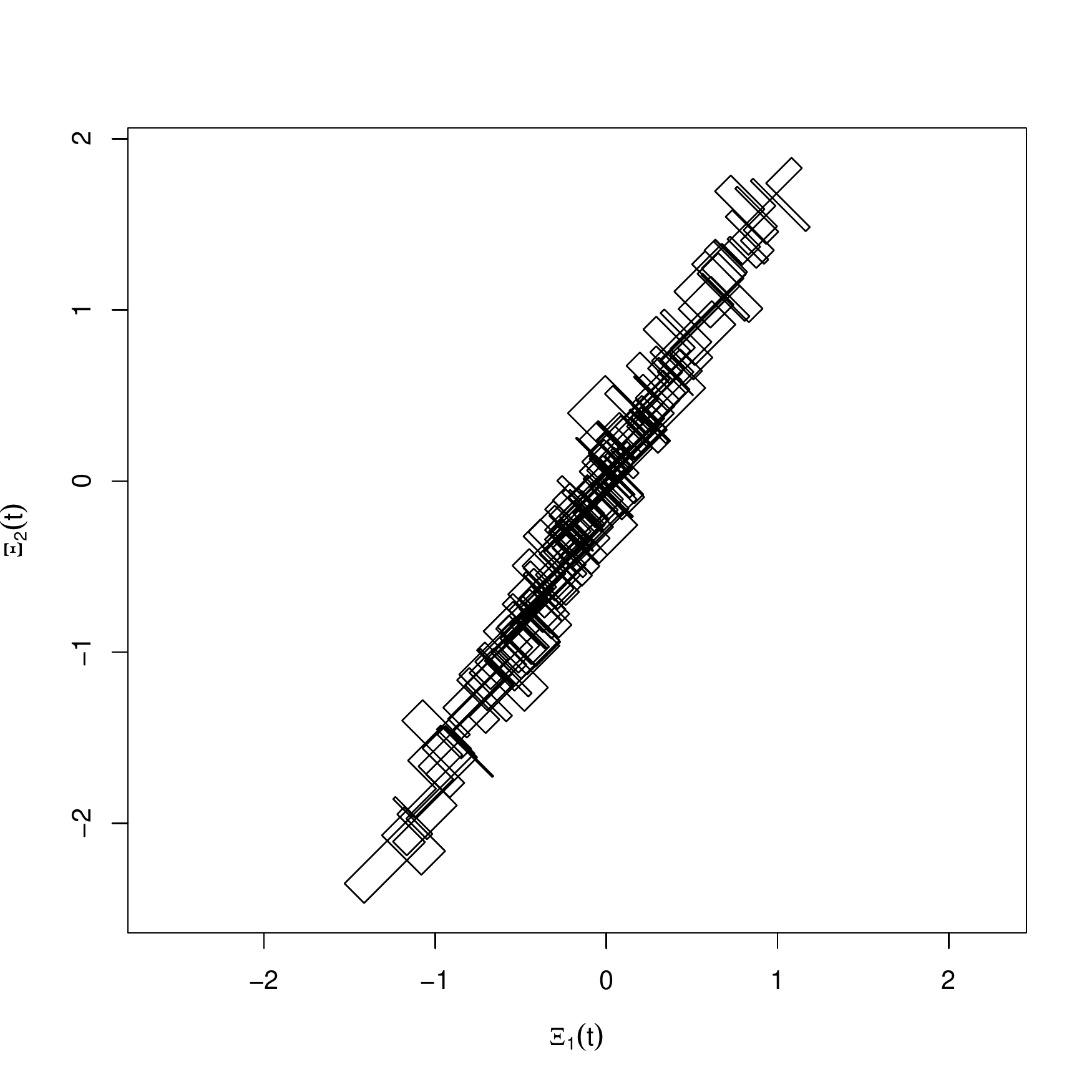}
 \caption{2D anisotropic Gaussian}
\end{subfigure}
\begin{subfigure}[b]{0.45 \textwidth}
 \includegraphics[width=\textwidth]{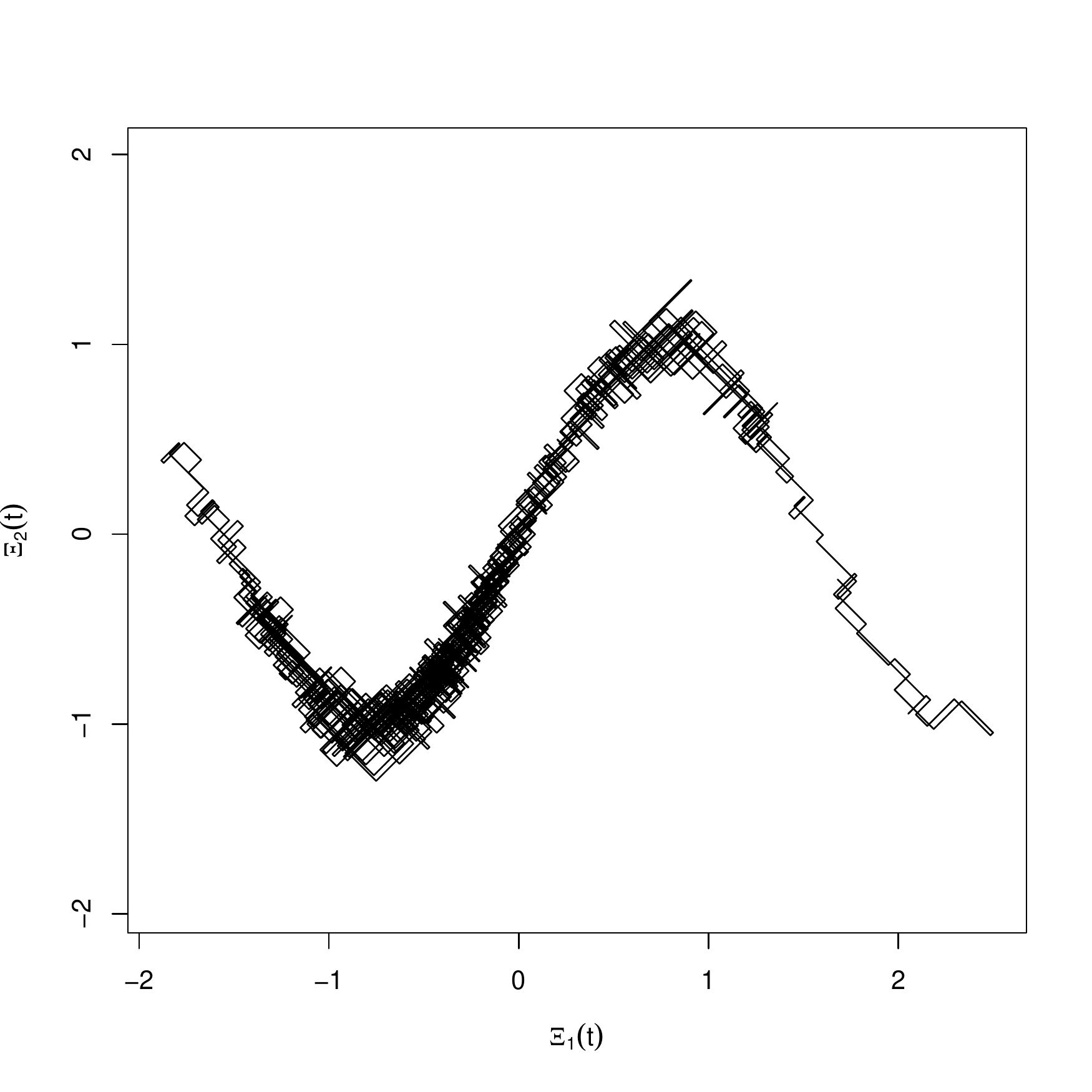}
 \caption{2D S-shaped density}
\end{subfigure}
\begin{subfigure}[b]{0.45 \textwidth}
 \includegraphics[width=\textwidth]{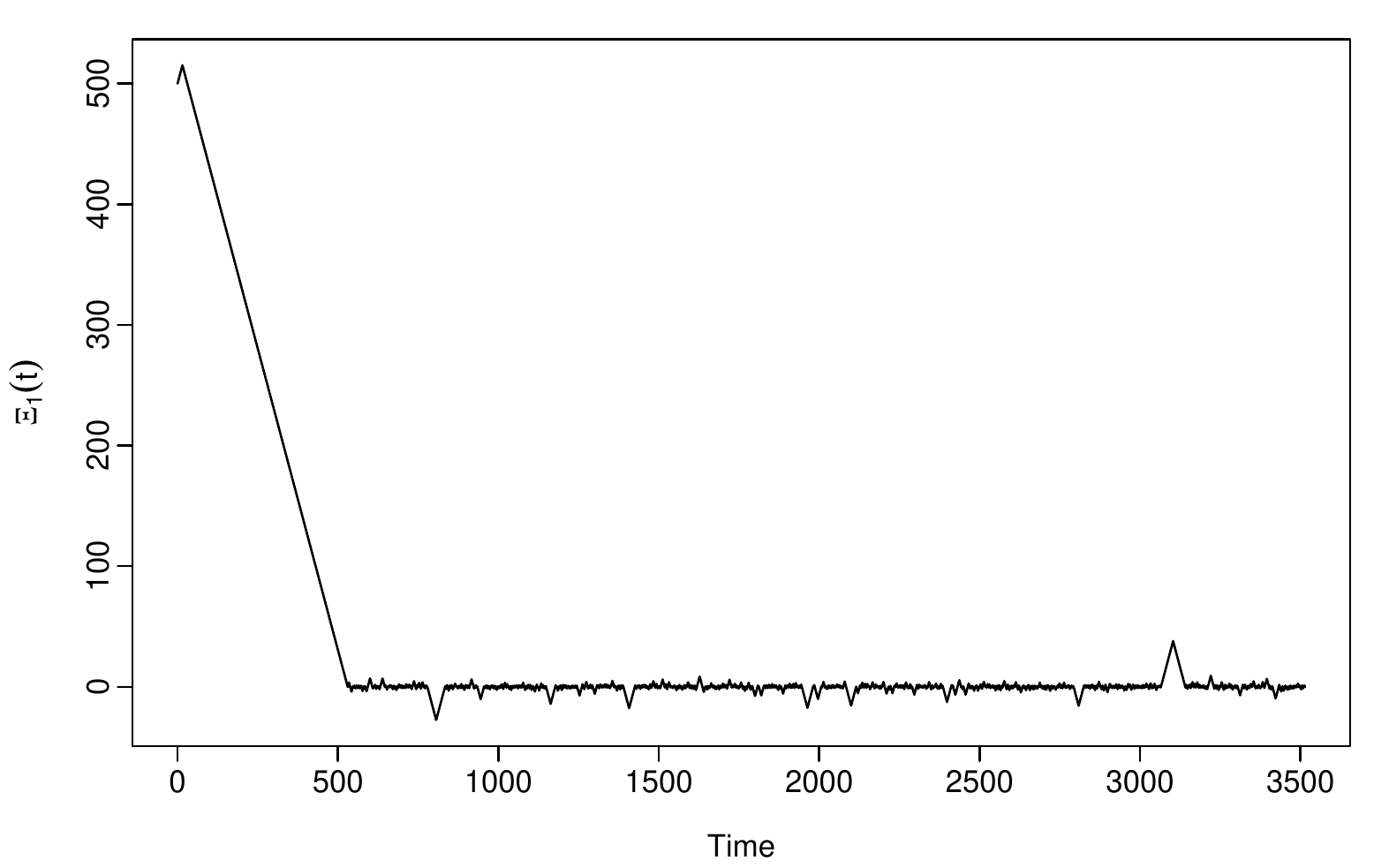}
 \caption{1D Cauchy}
\end{subfigure}
\begin{subfigure}[b]{0.45 \textwidth}
 \includegraphics[width=\textwidth]{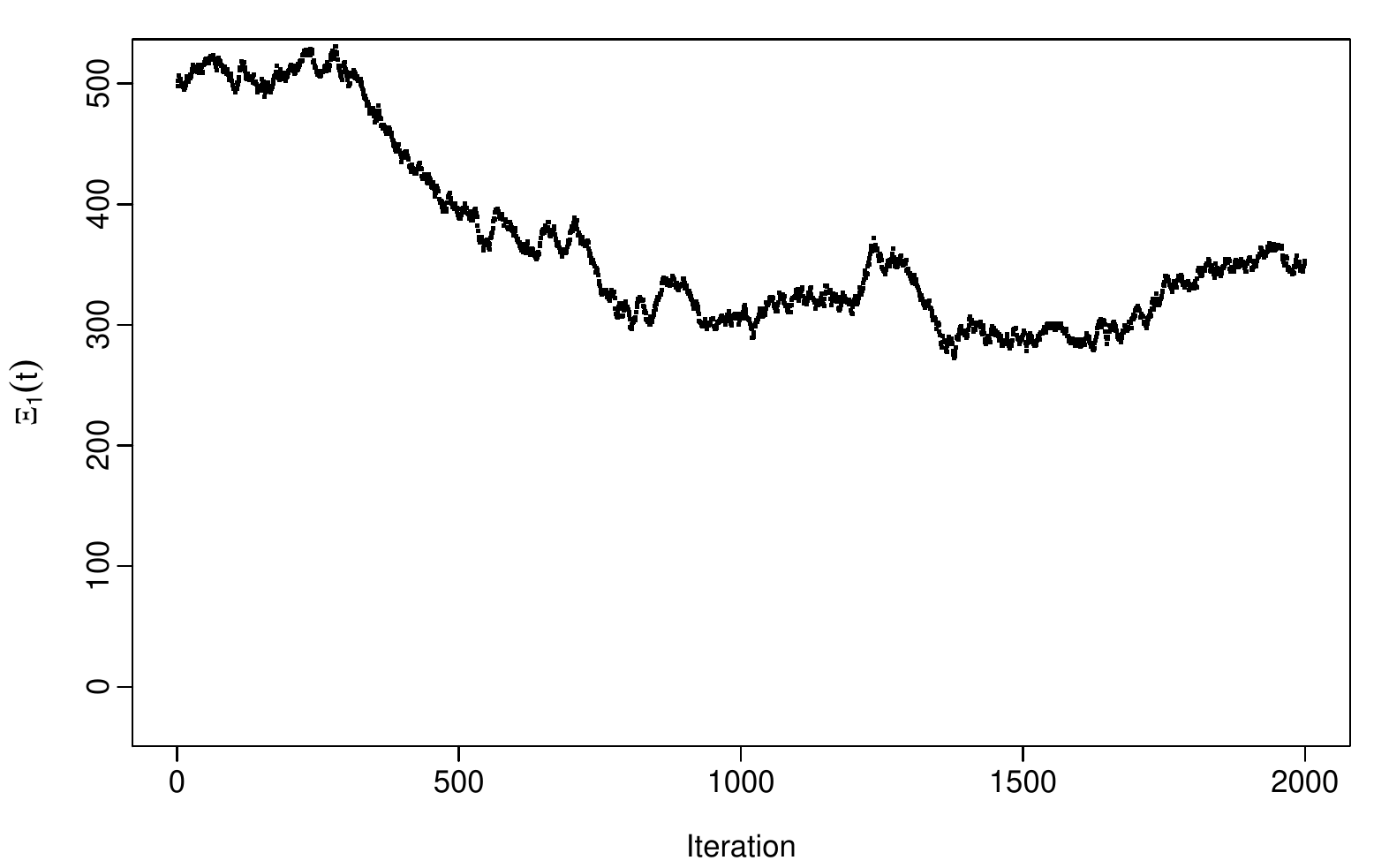}
 \caption{1D Cauchy -- MALA}
\end{subfigure}

\caption{Top two rows: example trajectories of the canonical Zig-Zag process. In (a) and (b) the horizontal axis shows time and the vertical axis the $\Xi$-coordinate of the 1D process. In (c) and (d), 
the trajectories in $\R^2$ of $(\Xi_1, \Xi_2)$ are plotted. Bottom row: Zig-Zag process (e) and MALA (f) for a Cauchy target with both processes started in the tail.
\label{fig:basic-examples} }
\end{figure}
\subsection{Invariant distribution}
\label{sec:invariant-dist}

The most important aspect of the Zig-Zag process is that in many cases the switching rates are directly related to an easily identifiable invariant distribution. Let $C^1(\R^d)$ denote the space of continuously differentiable functions on $\R^d$.
For $\theta \in \{-1,+1\}^d$ and $i \in \{1,\dots, d\}$, let $F_i[\theta] \in \{-1,+1\}^d$ denote the binary vector obtained by flipping the $i$-th component of $\theta$; i.e.
\[ (F_i[\theta])_j = \begin{cases} \theta_j \quad & \mbox{if} \ i \neq j, \\
                      - \theta_j \quad & \mbox{if} \ i = j.
                     \end{cases}.\]

We introduce the following assumption.

\begin{assumption}
\label{ass:invariant-measure}
For some function $\Psi \in C^1(\R^d)$ satisfying 
\begin{equation} \label{eq:integrability} \int_{\R^d} \exp(- \Psi(\xi)) \ d \xi < \infty\end{equation} we have
\begin{equation} \label{eq:relation_lambda} \lambda_i(\xi,\theta) - \lambda_i(\xi,F_i[\theta])  = \theta_i \partial_i \Psi(\xi) \quad \mbox{for all} \ (\xi,\theta) \in E,  i = 1, \dots, d.\end{equation}
\end{assumption}

Throughout this paper we will refer to $\Psi$ as the \emph{negative log density}.
Let $\mu_0$ denote the measure on $\mathcal B(E)$ such that, for $A \in \mathcal B(\R^d)$ and $\theta \in \{-1,+1\}^d$,
$\mu_0(A \times \{\theta\}) = \mathrm{Leb}(A)$,
with $\mathrm{Leb}$ denoting Lebesgue measure on $\R^d$.
\begin{theorem}
\label{thm:invariant_measure}
Suppose Assumption~\ref{ass:invariant-measure} holds. Let $\mu$ denote the probability distribution on $E$ such that $\mu$ has Radon-Nikodym derivative 
\begin{equation} \label{eq:inv-measure} \frac{d \mu}{d \mu_0}(\xi, \theta) = \frac {\exp(-\Psi(\xi))}{Z}, \quad (\xi, \theta) \in E,\end{equation}
where $Z = \int_E \exp(-\Psi) \ d \mu_0$.
Then the Zig-Zag process $(\Xi, \Theta)$ with switching rates $(\lambda_i)_{i=1}^d$ has invariant distribution $\mu$.
\end{theorem}

The proof is located in the Section 1 of the Supplementary Material. 
%
We see that under the invariant distribution of the Zig-Zag process, $\xi$ and $\theta$ are independent, with $\xi$ having density proportional to $\exp(-\Psi(\xi))$ and $\theta$ having a uniform distribution on the points in $\{-1,+1\}^d$.

%


For $a \in \R$, let $(a)^+ := \max(0, a)$ and $(a)^- := \max(0,-a)$ denote the positive and negative parts of $a$, respectively. We will often use the trivial identity $a = (a)^+ - (a)^-$ without comment. The following result characterizes the switching rates for which~\eqref{eq:relation_lambda} holds.

\begin{proposition}
\label{prop:characterization_lambda}
Suppose $\lambda : E \rightarrow \R^d_+$ is continuous. Then Assumption~\ref{ass:invariant-measure} is satisfied if and only if there exists a continuous function $\gamma : E \rightarrow \R_+^d$ such that for all $i = 1, \dots, d$ and $(\xi,\theta) \in E$, $\gamma_i(\xi,\theta) = \gamma_i(\xi, F_i [\theta])$ and, for $\Psi \in C^1(\R^d)$ satisfying~\eqref{eq:integrability},
\begin{equation} \label{eq:characterization_lambda} \lambda_i(\xi,\theta) = \left( \theta_i \partial_i \Psi(\xi)\right)^+ + \gamma_i(\xi,\theta).\end{equation}
\end{proposition}

The proof is located in Section 1 of the Supplementary Material. 

%


\subsection{Zig-Zag process for Bayesian inference}
\label{sec:bayesian-inference}

One application of the Zig-Zag process is as an alternative to MCMC for sampling from posterior distributions in Bayesian statistics. We show here that it is straightforward to derive a class of Zig-Zag processes
that have a given posterior distribution as their invariant distribution. The dynamics of the Zig-Zag process only depend on knowing the posterior density up to a constant of proportionality. 

To keep notation consistent with that used for the Zig-Zag process, let $\xi \in \R^d$ denote a vector of continuous parameters. We are given a prior density function for $\xi$, which we denote by $\pi_0(\xi)$, and observations
$x^{1:n}=(x^1,\ldots,x^n)$. Our model for the data defines a likelihood function $L(x^{1:n}|\xi)$. Thus the posterior density function is 
\[
 \pi(\xi) \propto \pi_0(\xi)L(x^{1:n}|\xi).
\]
We can write $\pi(\xi)$ in the form of the previous section, 
\[
 \pi(\xi) = \frac 1 Z \exp(-\Psi(\xi)), \quad \xi \in \R^d,
\]
where $\Psi(\xi)=-\log\pi_0(\xi)-\log L(x^{1:n}|\xi)$, and $Z = \int_{\R^d} \exp(-\Psi(\xi)) \ d \xi$ is the unknown normalising constant. 
Now assuming that $\log \pi_0(\xi)$ and $\log L(x^{1:n}|\xi)$ are both continuously differentiable with respect to $\xi$, 
from~\eqref{eq:characterization_lambda} 
a Zig-Zag process with rates
\[
 \lambda_i(\xi,\theta)= \left(\theta_i\partial_i \Psi(\xi)\right)^+
\]
will have the posterior density $\pi(\xi)$ as the marginal of its invariant distribution. 
We call the 
process with these rates the \emph{Canonical Zig-Zag process} for the negative log density $\Psi$. As explained in
Proposition~\ref{prop:characterization_lambda}, we can construct a family of Zig-Zag processes with the same invariant distribution by choosing any set of functions $\gamma_i(\xi,\theta)$, for $i=1,\ldots,d$, which
take non-negative values and for which
$\gamma_i(\xi,\theta)=\gamma_i(\xi,F_i[\theta])$, and setting
\[
 \lambda_i(\xi,\theta)= \left(\theta_i\partial_i \Psi(\xi)\right)^+ +  \gamma_i(\xi,\theta), \mbox{ for $i=1,\ldots,d$.}
\]
The intuition here is that $\lambda_i(\xi,\theta)$ is the rate at which we transition from $\theta$ to $F_i[\theta]$. The condition $\gamma_i(\xi,\theta)=\gamma_i(\xi,F_i[\theta])$ means that we increase
by the same amount both the rate at which we will transition from $\theta$ to $F_i[\theta]$ and vice versa. 
As our invariant distribution places the same probability of
being in a state with velocity $\theta$ as that of being in state $F_i[\theta]$, these two changes in rate cancel out in terms of their effect on the invariant distribution. Changing the rates in
this way does impact the dynamics of the process, with larger $\gamma_i$ values corresponding to more frequent changes in the velocity of the Zig-Zag process, and we would expect the resulting process
to mix more slowly. 

Under the assumption that the Zig-Zag process has the desired invariant distribution and 
is ergodic, it follows from the Birkhoff ergodic theorem that for any bounded continuous function $f : E \rightarrow \R$,
\[ \lim_{t \rightarrow \infty} \frac 1 t \int_0^t f(\Xi(s), \Theta(s)) \ d s = \int_E f \ d \mu,\]
for any initial condition $(\xi, \theta) \in E$. Sufficient conditions for ergodicity will be discussed in the following section. Taking $\gamma$ to be positive and bounded everywhere ensures ergodicity, as will be established in Theorem~\ref{thm:positive-irreducible}.

\subsection{Ergodicity of the Zig-Zag process}
\label{sec:ergodicity}

We have established in Section~\ref{sec:invariant-dist} that for any continuously differentiable, positive density $\pi$ on $\R^d$ a Zig-Zag process can be constructed that has $\pi$ as its marginal stationary density.
 In order for ergodic averages $\frac 1 T \int_0^T f(\Xi(s)) \ d s$ of the Zig-Zag process to converge asymptotically to $\pi(f)$, we further require $(\Xi(t), \Theta(t))$ to be ergodic, i.e. to admit a \emph{unique} invariant distribution. 

Ergodicity is directly related to the requirement that $(\Xi(t), \Theta(t))$ is irreducible, i.e. the state space is not reducible into components which are each invariant for the process $(\Xi(t), \Theta(t))$. 
For the one-dimensional Zig-Zag process, (exponential) ergodicity has already been established under mild conditions \cite[]{BierkensRoberts2015}. As we discuss below, irreducibility, and thus ergodicity, can be established for large classes of multi-dimensional target distributions, such as i.i.d. Gaussian distributions, and also if the switching rates $\lambda_i(\xi,\theta)$ are positive for all $i = 1, \dots, d$, and $(\xi,\theta) \in E$.

Let $P^t((\xi,\theta), \cdot)$ be the transition kernel of the Zig-Zag process with initial condition $(\xi,\theta)$.
A function $f : E \rightarrow \R$ is called \emph{norm-like} if $\lim_{\|\xi\| \rightarrow \infty} f(\xi, \theta) = \infty$ for all $\theta \in \{-1,+1\}^d$.
Let $\|\cdot\|_{\mathrm{TV}}$ denote the total variation norm on the space of signed measures.
First we consider the one-dimensional case.

\begin{assumption}
\label{ass:exp-ergodicity-1d}
Suppose $d = 1$ and there exists $\xi_0 > 0$ such that 
\begin{itemize}
 \item[(i)] $\inf_{\xi \geq \xi_0} \lambda(\xi,+1) >  \sup_{\xi \geq \xi_0} \lambda(\xi,-1)$, and
 \item[(ii)]  $\inf_{\xi \leq -\xi_0} \lambda(\xi,-1) > \sup_{\xi \leq -\xi_0} \lambda(\xi,+1)$.
\end{itemize}
\end{assumption}

\begin{proposition}{\cite[Theorem 5]{BierkensRoberts2015}}
\label{prop:ergodicity-1d}
Suppose Assumption~\ref{ass:exp-ergodicity-1d} holds. Then there exists a function $f : E \rightarrow [1,\infty)$ which is norm-like 
such that the Zig-Zag process is $f$-exponentially ergodic, i.e. there exists a constant $\kappa > 0$ and $0 < \rho < 1$ such that
\[ \| P^t((\xi,\theta), \cdot) - \pi\|_{\mathrm{TV}} \leq \kappa f(\xi,\theta)  \rho^t \quad \mbox{for all $(\xi, \theta) \in E$ and $t \geq 0$.}\]
\end{proposition}

\begin{example}
\label{ex:gaussian-ergodic}
As an example of fundamental importance, which will also be used in the proof of Theorem~\ref{thm:positive-irreducible}, consider a one-dimensional Gaussian distribution. For simplicity let $\pi(\xi)$ be centred, $\pi(\xi) = \frac 1 { \sqrt{ 2 \pi \sigma^2}} \exp\left( - \frac{\xi^2}{2 \sigma^2} \right)$ for some $\sigma > 0$.
According to~\eqref{eq:characterization_lambda} the switching rates take the form
\[ \lambda(\xi,\theta) = \left( \theta \xi/\sigma^2 \right)^+ + \gamma(\xi), \quad (\xi, \theta) \in E.\]
As long as $\gamma$ is bounded from above, Assumption~\ref{ass:exp-ergodicity-1d} is satisfied. In particular this holds if $\gamma$ is equal to a non-negative constant.
\end{example}

\begin{remark}
We say a probability density function $\pi$ is of \emph{product form} if $\pi(\xi) = \prod_{i=1}^d \pi_i(\xi_i)$, where $\pi_i : \R^d \rightarrow (0,\infty)$ are one-dimensional probability density functions. 
When its target density is of product form the Zig-Zag process is the concatenation of independent Zig-Zag processes. In this case the negative log density is of the form $\Psi(\xi) = \sum_{i=1}^d \Psi_i(\xi_i)$ and the switching rate for the $i$-th component of $\theta$ is
\begin{equation} \label{eq:lambda_product_form} \lambda_i(\xi,\theta) = \left( \theta_i \Psi_i'(\xi_i) \right)^+ + \gamma_i(\xi).\end{equation}
As long as $\gamma_i(\xi) = \gamma_i(\xi_i)$, i.e. if $\gamma_i(\xi)$ only depends on the $i$-th coordinate of $\xi$, the switching rate of coordinate $i$ is independent of the other coordinates $\xi_j$, $j \neq i$. It follows that the switches of the $i$-th coordinate can be generated by a one-dimensional time inhomogeneous Poisson process, which is independent of the switches in the other coordinates. As a consequence the $d$-dimensional Zig-Zag process $(\Xi(t), \Theta(t)) = (\Xi^1(t), \dots, \Xi_d(t), \Theta^1(t), \dots, \Theta^d(t))$ consists of a combination of $d$ independent Zig-Zag processes $(\Xi^i(t), \Theta^i(t))$, $i = 1, \dots, d$.
\end{remark}

Suppose $P(x, dy)$ is the transition kernel of a Markov chain on a state space $E$. We say that the Markov chain associated to $P$ is \emph{mixing} if there exists a probability distribution $\pi$ on $E$ such that 
\[ \lim_{k \rightarrow \infty} \| P^k (x,\cdot) - \pi \|_{\mathrm{TV}} = 0 \quad \mbox{for all $x \in E$.}\]
For any continuous time Markov process with family of transition kernels $P^t(x,dy)$ we can consider the associated \emph{time-discretized process}, which is a Markov chain with transition kernel $Q(x, dy) :=  P^{\delta}(x, dy)$ for a fixed $\delta > 0$. The value of $\delta$ will be of no significance in our use of this construction.

\begin{proposition}
\label{prop:product-mixing}
Suppose $\pi$ is of product form and $\lambda : E \rightarrow \R_+^d$ admits the representation~\eqref{eq:lambda_product_form} with $\gamma_i(\xi)$ only depending on $\{\xi_i, i = 1, \dots, d\}$. Furthermore suppose that for every $i = 1, \dots, d$, the one-dimensional time-discretized Zig-Zag process corresponding to switching rate $\lambda_i$ is mixing in $\R \times \{-1, +1\}$. Then the time-discretized $d$-dimensional Zig-Zag process with switching rates $(\lambda_i)$ is mixing. In particular, the multi-dimensional Zig-Zag process admits a unique invariant distribution. 
\end{proposition}

\begin{proof}
This follows from the decomposition of the $d$-dimensional Zig-Zag process as $d$ one-dimensional Zig-Zag processes and Lemma 1.1 in the Supplementary material.
\end{proof}

\begin{example}
\label{ex:gaussian-multidimensional}
Continuing  Example~\ref{ex:gaussian-ergodic}, consider the simple case in which $\pi$ is of product form with each $\pi_i$ a centered Gaussian density function with variance $\sigma_i^2$. It follows from Proposition~\ref{prop:product-mixing} and Example~\ref{ex:gaussian-ergodic} that the multi-dimensional canonical Zig-Zag process (i.e. the Zig-Zag process with $\gamma_i \equiv 0$) is mixing. 
This is different from the Bouncy Particle Sampler \cite[]{Bouchard-Cote2015}, which is not ergodic for an i.i.d. Gaussian without `refreshments' of the momentum variable. 
\end{example}

We now show that strict positivity of the rates ensures ergodicity.


\begin{theorem}
\label{thm:positive-irreducible}
Suppose $\lambda : E \rightarrow (0,\infty)^d$, in particular $\lambda_i(\xi,\theta)$ is positive for all $i = 1, \dots,d$ and $(\xi,\theta) \in E$. Then there exists at most a single invariant measure for the Zig-Zag process with switching rate $\lambda$.
\end{theorem}

The proof of this result consists of a Girsanov change of measure with respect to a Zig-Zag process targeting an i.i.d. standard normal distribution, which we know to be irreducible. The irreducibility then carries over to the Zig-Zag process with the stated switching rates. A detailed proof can be found in the Supplementary material. 

\begin{remark}
Based on numerous experiments, we conjecture that the canonical multi-dimensional Zig-Zag process is ergodic in general under only mild conditions.
A detailed investigation of ergodicity  will be the subject of a forthcoming paper \cite[]{Bierkens2017}.
\end{remark}

\section{Implementation}
\label{sec:sim}

As mentioned earlier, the main computational challenge is an efficient simulation of the random times $T^k_i$ introduced in Section~\ref{sec:construction}. We will focus on simulation by means of Poisson thinning.

\begin{proposition}[Poisson thinning, {\cite{LewisShedler1979}}]
\label{prop:poisson-thinning}
Let $m : \R_+ \rightarrow \R_+$ and $M : \R_+ \rightarrow \R_+$ be continuous such that $m(t) \leq M(t)$ for $t \geq 0$. Let $\tau^1, \tau^2, \dots$ be the increasing finite or infinite sequence of points of a Poisson
process with rate function $(M(t))_{t \geq 0}$. For all $i$, delete the point $ \tau^i$ with probability $1 - m(\tau^i) / M(\tau^i)$. Then the remaining points, $\widetilde \tau^1, \widetilde \tau^2, \dots$ say, form a non-homogeneous Poisson process  with rate function $(m(t))_{t \geq 0}$.
\end{proposition}


Now for a given initial point $(\xi, \theta) \in E$, let $m_i(t) := \lambda_i(\xi + \theta t, \theta)$, for $i = 1, \dots, d$, and suppose we have available continuous functions $M_i(t)$ such that $m_i(t) \leq M_i(t)$ for $i = 1, \dots, d$ and $t \geq 0$. We call these $(M_i)_{i=1}^d$ \emph{computational bounds} for $(m_i)_{i=1}^d$. We can use Proposition~\ref{prop:poisson-thinning} to obtain the first switching times $(\widetilde \tau^1_i)_{i=1}^d$ from a (theoretically infinite) collection of \emph{proposed switching times} $(\tau^1_i, \tau^2_i, \dots)_{i=1}^d$ given the initial point $(\xi, \theta)$, and use the obtained skeleton point at time $\widetilde \tau^1 := \min_{i \in \{1, \dots, d \}} \widetilde \tau_i^1$ as a new initial point (which is allowed by the strong Markov property) with the component $i_0 = \argmin_{i \in \{1, \dots, d \}} \widetilde \tau_i^1$ of $\theta$ switched. 

The strong Markov property of the Zig-Zag process simplifies the computational procedure further: we can draw for each component $i = 1, \dots, d$ the first proposed switching time $\tau_i := \tau_i^1$, determine $i_0 := \argmin_{i  \in \{1, \dots, d\}} \tau_i$ and decide whether the appropriate component of $\theta$ is switched at this time with probability $m_{i_0}(\tau)/M_{i_0}(\tau)$, where $\tau := \tau_{i_0}$. Then since $\tau$ is a stopping time for the Markov process, we can use the obtained point of the Zig-Zag process at time $\tau$ as new starting point, regardless of whether we switch a component of $\theta$ at the obtained skeleton point. A full computational procedure for simulating the Zig-Zag process is given by Algorithm~\ref{alg:general}.

\begin{algorithm}
Input: initial condition $(\xi, \theta) \in E$. \\
Output: a sequence of skeleton points $(T^k, \Xi^k, \Theta^k)_{k=0}^{\infty}$.

\begin{enumerate}
 \item $(T^0, \Xi^0, \Theta^0) := (0, \xi, \theta)$.
 \item for $k = 1, 2, \dots$
 \begin{enumerate}
 \item Define $m_i(t) := \lambda_i(\Xi^{k-1} + \Theta^{k-1} t, \Theta^{k-1})$ for $t \geq 0$ and $i = 1, \dots, d$.
 \item For $i = 1, \dots, d$, let $(M_i)$ denote computational bounds for $(m_i)$.
 \item Draw $\tau_1, \dots, \tau_d$ such that $\P \left(\tau_i \geq t \right) = \exp \left( - \int_0^t M_i(s) \ d s \right)$.
 \item $i_0 := \argmin_{i=1, \dots, d} \{ \tau_i\}$ and $\tau := \tau_{i_0}$.
 \item $(T^k, \Xi^k) := (T^{k-1} + \tau, \Xi^{k-1} + \Theta^{k-1} \tau)$
 \item With probability $m_{i_0}(\tau)/M_{i_0}(\tau)$,
 \begin{itemize}
 \item $\Theta^{k} := F_{i_0} [\Theta^{k-1}]$,
 \end{itemize}
 otherwise
 \begin{itemize}
 \item $\Theta^{k} := \Theta^{k-1}$.
 \end{itemize}
\end{enumerate}
\end{enumerate}
\caption{Zig-Zag Sampling (ZZ)}
\label{alg:general}
\end{algorithm}

\subsection{Computational bounds}
\label{sec:compbd}

We now come to the important issue of obtaining computational bounds for the Zig-Zag Process, i.e. useful upper bounds for the switching rates $(m_i)$. If we can compute the inverse function $G_i(y) := \inf \{t \geq 0 : H_i(t) \geq y\}$ of $H_i : t \mapsto \int_0^t M_i(s) \ d s$, we can simulate $\tau_1, \dots, \tau_d$ using the CDF inversion technique, i.e. by drawing i.i.d. uniform random variables $U_1, \dots, U_d$ and setting $\tau_i := G_i(- \log U_i)$, $i =1, \dots d$. 

Let us ignore the subscript $i$ for a moment. Examples of computational bounds are piecewise affine bounds of the form $M: t \mapsto ( a + b t)^+$, with $a, b \in \R$, 
and the constant bounds $M : t \mapsto c$ for $c \geq 0$. 
It is also possible to simulate using the combined rate $M: t \mapsto \min(c, (a +bt)^+)$. In these cases, $H(t) = \int_0^t M(s) \ d s$ is piecewise linear or quadratic and non-decreasing, so we can obtain an explicit expression for the inverse function, $G$. 

The computational bounds are directly related to the algorithmic efficiency of Zig-Zag Sampling. From Algorithm~\ref{alg:general}, it is clear that for every simulated time $\tau$ a single component of $\lambda$ needs to be evaluated, which corresponds by~\eqref{eq:characterization_lambda} to the evaluation of a single component of the gradient of the negative log density $\Psi$. The magnitude of the computational bounds, $(M_i)$, will 
determine how far the Zig-Zag process will have moved in the state space before a new evaluation of a component of $\lambda$ is required, and 
we will pay close attention to the scaling of $M_i$ with respect to the number of available observations in a Bayesian inference setting.

\subsection{Example: globally bounded log density gradient}
If there are constants $c_i > 0$ such that $\sup_{\xi \in \R^d} |\partial_i \Psi(\xi)| \leq c_i$, $i =1, \dots d$, then we can use the global upper bounds $M_i(t) = c_i$ for $t \geq 0$. Indeed, for $(\xi, \theta) \in E$, 
\[ \lambda_i(\xi, \theta) = \left(\theta_i \partial_i \Psi(\xi)\right)^+ \leq |\partial_i \Psi(\xi)| \leq c_i.\]
Algorithm~\ref{alg:general} may be used with $M_i \equiv c_i$ for $i = 1, \dots, d$ at every iteration.

This 
situation arises with heavy-tailed distributions. E.g. if $\pi$ is Cauchy, then $\Psi(\xi) = \log(1+\xi^2)$, and consequently $\lambda(\xi, \theta) = \left(\frac{2 \theta \xi}{1+\xi^2} \right)^+ \leq 1$.

\subsection{Example: negative log density with dominated Hessian}
\label{sec:sampling-dominated-hessian}


Another important case is when there exists a positive definite matrix $Q \in \R^{d \times d}$ which dominates the Hessian $H_{\Psi}(\xi)$ in the positive definite ordering of matrices for every $\xi \in \R^d$. Here $H_{\Psi}(\xi) = (\partial_i \partial_j \Psi(\xi))_{i,j=1}^d$ denotes the Hessian of $\Psi$. 


Denote the Euclidean inner product in $\R^d$ by $\langle \cdot, \cdot\rangle$. For $p \in [1,\infty]$ the $\ell^p$-norm on $\R^d$ and the induced matrix norms are both denoted by $\|\cdot\|_p$. For symmetric matrices $S, T \in \R^{d \times d}$ we write $S \preceq T$ if $\langle v, S v \rangle \leq \langle v, T v \rangle$ for every $v \in \R^d$, or in words, if $T$ dominates $S$ in the positive definite ordering. The key assumption is that $H_{\Psi}(\xi) \preceq Q$ for all $\xi \in \R^d$, where $Q \in R^{d \times d}$ is positive definite.  In particular, if $\|H_{\Psi}(\xi)\|_2 \leq c$ for all $\xi$, then this holds for $Q = c I$.
We let $(e_i)_{i=1}^d$ denote the canonical basis vectors in $\R^d$. 

For an initial value $(\xi, \theta) \in E$, we move along the trajectory $t \mapsto \xi(t):= \xi + \theta t$.
Let $a_i$ denote an upper bound for $\theta_i \partial_i \Psi(\xi)$, $i =1, \dots, d$ and let $b_i := \sqrt{d} \|Q e_i\|_2$.
For general symmetric matrices $S, T$ with $S \preceq T$, we have for any $v, w \in \R^d$ that 
\begin{equation} \label{eq:inner-product-inequality} \langle v, S w \rangle \leq \|v\|_2 \|Sw\|_2 \leq \|v\|_2 \|T w \|_2.\end{equation}
Applying this inequality we obtain for $i = 1, \dots, d$,
\begin{align*}
 \theta_i \partial_i \Psi(\xi(t)) & = \theta_i \partial_i \Psi(\xi)+ \int_0^t \sum_{j=1}^d \partial_i \partial_j \Psi(\xi(s)) \theta_j \ d s \leq a_i + \int_0^t \langle H_{\Psi}(\xi(s)) e_i, \theta \rangle \ d s \\
 & \leq a_i + \int_0^t \|Q e_i\|_2 \|\theta\|_2 \ d s = a_i + b_i t.
\end{align*}
It thus follows that 
\[ \lambda_i(\xi(t), \theta) = \left(\theta_i \partial_i \Psi(\xi(t))\right)^+ \leq (a_i + b_i t)^+.\]
Hence the general Zig-Zag Algorithm may be applied taking 
\[ M_i(t) := (a_i + b_i t)^+, \quad t \geq 0, \quad i =1, \dots, d,\]
with $a_i$ and $b_i$ as specified above.
A complete procedure for Zig-Zag Sampling for a log density with dominated Hessian is  provided in Algorithm~\ref{alg:dominated-hessian}.

\begin{algorithm}
Input: initial condition $(\xi, \theta) \in E$. \\
Output: a sequence of skeleton points $(T^k, \Xi^k, \Theta^k)_{k=0}^{\infty}$. \\
\begin{enumerate}
 \item $(T^0, \Xi^0, \Theta^0) := (0, \xi, \theta)$.
 \item $a_i := \theta_i \partial_i \Psi(\xi)$, $i = 1, \dots, d$.
 \item $b_i := Q e_i \sqrt{d}$, $i=1,\dots, d$.
 \item For $k = 1,2, \dots$
 \begin{enumerate}
  \item Draw $\tau_i$ such that $\P(\tau_i \geq t) = \exp \left( -\int_0^t (a_i + b_i s)^+ \ d s \right)$, $i = 1, \dots, d$.
  \item $i_0 := \argmin_{i \in \{1, \dots, d\}} \tau_i $ and  $\tau := \tau_{i_0}$.
 \item $(T^k, \Xi^k, \Theta^k) := (T^{k-1} + \tau, \Xi^{k-1} + \Theta^{k-1} \tau, \Theta^{k-1})$
   \item $a_i := a_i + b_i \tau$, $i = 1, \dots, d$.
  \item with probability $\frac{\left(\Theta^{k-1}_{i_0} \partial_{i_0} \Psi(\Xi^k)\right)^+}{\left(a_{i_0} \right)^+}$,
  \begin{itemize}
   \item $\Theta^{k} := F_{i_0} [\Theta^{k-1}]$ 
  \end{itemize}
  otherwise
  \begin{itemize}
   \item $\Theta^k := \Theta^{k-1}$.
  \end{itemize}
   \item $a_{i_0} := \Theta^{k-1}_{i_0} \partial_{i_0} \Psi(\Xi^k)$ (re-using the earlier computation)
\end{enumerate}
\end{enumerate}
\caption{Zig-Zag Sampling for log density with dominated Hessian}
\label{alg:dominated-hessian}
\end{algorithm}

\begin{remark}
It is also possibly to apply inequality~\eqref{eq:inner-product-inequality} in such a way as to obtain the estimate 
\[ \langle H_{\Psi}(\xi(s)) e_i, \theta \rangle = \langle e_i, H_{\Psi}(\xi(s)) \theta\rangle \leq \|e_i \|_2 \| Q \theta\|_2 = \| Q \theta\|_2.\]
This requires us to compute $Q \theta$ whenever $\theta$ changes (a computation of $O(d)$).
%
\end{remark}

\section{Big data Bayesian inference by means of error-free sub-sampling}
\label{sec:big}

Throughout this section we assume the derivatives of $\Psi$ admit the representation 
\begin{equation} \label{eq:grad-logdensity-as-average} \partial_i \Psi(\xi) = \frac 1 n \sum_{j=1}^n E^{j}_i(\xi), \quad i = 1, \dots, d, \quad \xi \in \R^d, \end{equation}
with $(E^j)_{j=1}^n$ continuous functions mapping $\R^d$ into $\R^d$.
The motivation for considering such a class of density functions is the problem of sampling from a posterior distribution for big data. 
The key feature of such posteriors is that they can be written as the product of a large number of terms.
For example consider the simplest example of this, where we have $n$ independent data points $(x^j)_{j=1}^n$ and for which the likelihood function is $L(\xi) =\prod_{j=1}^n f(x^j|\xi)$,
for some probability density or probability mass function $f$. In this case we can write the negative log density $\Psi$ associated with the posterior distribution as an average
\begin{equation} \label{eq:logdensity-as-average}
 \Psi(\xi)=\frac{1}{n}\sum_{j=1}^n \Psi^j(\xi), \quad \xi \in \R^d,
\end{equation}
where $\Psi^j(\xi)=-\log\pi_0(\xi)-n\log f(x^j|\xi)$, and we could choose $E_i^j(\xi)=\partial_i \Psi^j(\xi)$. It is crucial that every $E_i^j$ is a factor $O(n)$ cheaper to evaluate than the full derivative $\partial_i \Psi(\xi)$.

We will describe two successive improvements over the basic Zig-Zag Sampling (ZZ) algorithm specifically tailored to the situation in which~\eqref{eq:grad-logdensity-as-average} is satisfied. 
The first improvement consists of a sub-sampling approach where we need calculate only one of the $E_i^js$ at each simulated time, rather than sum of all $n$ of them. This sub-sampling approach (referred to as Zig-Zag with Sub-Sampling, ZZ-SS) comes at the cost of an increased computational bound. Our second improvement is to use control variates to reduce this bound, resulting in the Zig-Zag with Control Variates (ZZ-CV) algorithm.


\subsection{Main idea}
\label{sec:sub-sampling-main-idea}
Let $(\xi(t))_{t \geq 0}$ denote a linear trajectory originating in $(\xi,\theta) \in E$, i.e. $\xi(t) = \xi + \theta t$. Define a collection of switching rates along the trajectory $(\xi(t))$ by 
\[ m_i^j(t) := \left(\theta_i E_i^j(\xi(t))\right)^+, \quad i = 1, \dots, d, \quad j = 1, \dots, n, \quad t \geq 0.\]
We will make use of computational bounds $(M_i)$ as before, which this time bound $(m_i^j)$ uniformly. Let $M_i : \R_+ \rightarrow \R_+$ be continuous and satisfy
\begin{equation} \label{eq:comp-bound-sub-sampling} m_i^j(t) \leq M_i(t) \quad \mbox{for all $i = 1, \dots, d$, $j = 1, \dots, n$, and $t \geq 0$.}\end{equation}
We will generate random times according to the computational upper bounds $(M_i)$ as before. However, we now use a two-step approach to deciding whether to switch or not at the generated times. As before, for $i = 1, \dots, d$ let $(\tau_i)_{i=1}^d$ be simulated random times for which $\P(\tau \geq t) = \exp\left( - \int_0^t M_i(s) \ d s \right)$ and let $i_0 := \argmin_{i \in \{1, \dots, d \}} \tau_i$, and $\tau := \tau_{i_0}$. Then switch component $i_0$ of $\theta$ with probability $m_{i_0}^J(\tau)/ M_{i_0}(\tau)$, where $J \in \{1, \dots, n\}$ is drawn uniformly at random, independent of $\tau$. This `sub-sampling' procedure is detailed in Algorithm~\ref{alg:sub-sampling}. 
Depending on the choice of $E^j_i$, we will refer to this algorithm as Zig-Zag with Sub-Sampling (ZZ-SS, Section~\ref{sec:sub-sampling-global-bound}) or ZZ-CV
(Section~\ref{sec:sub-sampling-control-variates}).

\begin{theorem}
\label{thm:sub-sampling-rate}
Algorithm~\ref{alg:sub-sampling} generates a skeleton of a Zig-Zag process with switching rates given by
\begin{equation} \label{eq:lambda-sub-sampling} \lambda_i(\xi, \theta) = \frac 1 n \sum_{j=1}^n \left(\theta_i E_i^j(\xi)\right)^+,  \quad i = 1, \dots, d, \quad (\xi,\theta) \in E, \end{equation}
and invariant distribution $\mu$ given by~\eqref{eq:inv-measure}.\end{theorem}

\begin{proof}
Conditional on $\tau$, the probability that component $i_0$ of $\theta$ is switched at time $\tau$ is seen to be
\[ \mathbb E_J \left[m_{i_0}^J(\tau)/ M_{i_0}(\tau) \right] = \frac{\frac 1 n \sum_{j=1}^n m_{i_0}^j(\tau)}{M_{i_0}(T)} = \frac{m_{i_0}(\tau)}{M_{i_0}(\tau)},\]
where 
\[ m_i(t) := \frac 1 n \sum_{j=1}^n m_i^j(t) = \frac 1 n \sum_{j=1}^n \left(\theta_i E_i^j(\xi(t)) \right)^+, \quad i = 1, \dots, d, \quad t \geq 0.\]
By Proposition~\ref{prop:poisson-thinning} we thus have an effective switching rate $\lambda_i$ for switching the $i$-th component of $\theta$ given by~\eqref{eq:lambda-sub-sampling}. 
Finally we verify that the switching rates $(\lambda_i)$ given by~\eqref{eq:lambda-sub-sampling} satisfy~\eqref{eq:relation_lambda}. Indeed, 
\begin{align*} \lambda_i(\xi, \theta) - \lambda_i(\xi, F_i[\theta]) & = \frac 1 n \sum_{j=1}^n \left\{ \left( \theta_i E_i^j(\xi) \right)^+ -  \left( \theta_i E_i^j(\xi) \right)^- \right\} \\
& = \frac 1 n \sum_{j=1}^n \theta_i E_i^j(\xi) = \theta_i \partial_i \Psi(\xi). 
\end{align*}
By Theorem~\ref{thm:invariant_measure}, the Zig-Zag process has the stated invariant distribution.
\end{proof}

\begin{algorithm}[t]
Input: initial condition $(\xi, \theta) \in E$. \\
Output: a sequence of skeleton points $(T^k, \Xi^k, \Theta^k)_{k=0}^{\infty}$.

\begin{enumerate}
 \item $(T^0, \Xi^0, \Theta^0) := (0, \xi, \theta)$.
 \item for $k = 1, 2, \dots$
 \begin{enumerate}
 \item Define $m_i^j(t) := \left(\Theta^{k-1} E_i^j(\Xi^{k-1} + \Theta^{k-1} t) \right)^+$ for $t \geq 0$, $i = 1, \dots, d$ and $j = 1, \dots, n$.
 \item For $i = 1, \dots, d$, let $(M_i)$ denote computational bounds for $(m_i^j)$, i.e. satisfying~\eqref{eq:comp-bound-sub-sampling}.
 \item Draw $\tau_1, \dots, \tau_d$ such that $\P \left( \tau_i \geq t \right) = \exp \left( - \int_0^t M_i(s) \ d s \right)$.
 \item $i_0 := \argmin_{i=1, \dots, d} \tau_i$ and $\tau  := \tau_{i_0}$.
 \item $(T^k, \Xi^k) := (T^{k-1} + \tau, \Xi^{k-1} + \Theta^{k-1} \tau)$
 \item Draw $J \sim \mathrm{Uniform}(\{1, \dots, n\})$.
 \item With probability $m_{i_0}^J(\tau)/M_{i_0}(\tau)$,
 \begin{itemize}
 \item $\Theta^{k} := F_{i_0} [\Theta^{k-1}]$,
 \end{itemize}
 otherwise
 \begin{itemize}
 \item $\Theta^{k} := \Theta^{k-1}$.
 \end{itemize}
\end{enumerate}
\end{enumerate}
\caption{Zig-Zag with Sub-Sampling (ZZ-SS) / Zig-Zag with Control Variates (ZZ-CV)}
\label{alg:sub-sampling}
\end{algorithm}

The important advantage of using Zig-Zag in combination with sub-sampling is that at every iteration of the algorithm we only have to evaluate a single component of $E^j_i$, 
which reduces algorithmic complexity by a factor $O(n)$. However this may come at a cost. Firstly, the computational bounds $(M_i)$ may have to be increased
which in turn will increase the algorithmic complexity of simulating the Zig-Zag sampler.
Also, the dynamics of the Zig-Zag process will change, because the actual switching rates of the process are increased.
%
This increases the diffusivity of the continuous time Markov process, and affects the mixing properties in a negative way. 

\subsection{Zig-Zag with Sub-Sampling (ZZ-SS) for globally bounded log density gradient }
\label{sec:sub-sampling-global-bound}

A straightforward application of sub-sampling is possible if we have~\eqref{eq:logdensity-as-average}
with $\nabla \Psi^j$ globally bounded, i.e. there exist positive constants $(c_i)$ such that
\begin{equation} \label{eq:sub-sampling-global-bound} |\partial_i \Psi^j(\xi)| \leq c_i, \quad i = 1, \dots, d, \quad j = 1, \dots, n,  \quad \xi \in \R^d.\end{equation}
In this case we may take 
\[ E^j_i := \partial_i \Psi^j \quad \mbox{and} \quad  M_i(t) := c_i, \quad i = 1, \dots, d, \quad j = 1, \dots, n \quad  t \geq 0,\]
so that~\eqref{eq:comp-bound-sub-sampling} is satisfied. The corresponding version of Algorithm~\ref{alg:sub-sampling} will be called Zig-Zag with Sub-Sampling (ZZ-SS).

\subsection{Zig-Zag with Control Variates (ZZ-CV)}
\label{sec:sub-sampling-control-variates}

Suppose again that $\Psi$ admits the representation~\eqref{eq:logdensity-as-average}, and further suppose 
that the derivatives $(\partial_i \Psi^j)$ are globally and uniformly Lipschitz, i.e., there exist constants $(C_i)_{i=1}^n$ such that for some $p \in [1, \infty]$ and all $i = 1, \dots, d$, $j = 1, \dots, n$, and $\xi_1, \xi_2 \in \R^d$,
\begin{equation} \label{eq:sub-sampling-lipschitz-bound} \left| \partial_i \Psi^j(\xi_1) - \partial_i \Psi^j(\xi_2) \right| \leq C_i \| \xi_1 - \xi_2 \|_p.\end{equation}
To use these Lipschitz bounds we need to choose a reference point $\xi^{\star}$ in $\xi$-space, so that we can bound the derivative of the log density based on how close we are to this reference point. 
Now if we choose any fixed reference point, $\xi^{\star} \in \R^d$, we can use a control variate idea to write
\[
 \partial_i\Psi(\xi)=\partial_i\Psi(\xi^{\star})+\frac{1}{n}\sum_{i=1}^n \left[\partial_i\Psi^j(\xi)-\partial_i\Psi^j(\xi^{\star}) \right], \quad \xi \in \R^d, \quad i = 1, \dots, d.
\]
This suggests using
\[
E_i^j(\xi) := \partial_i\Psi(\xi^{\star})+ \partial_i \Psi^j(\xi) - \partial_i \Psi^j(\xi^{\star}), \quad \xi \in \R^d, \quad i = 1, \dots, d, \quad j = 1, \dots, n.
\]
The reason for defining $E_i^j(\xi)$ in this manner is to try and reduce its variability as we vary $j$. By the Lipschitz condition we have $E_i^j(\xi)\leq |\partial_i \Psi(\xi^{\star})| + C_i \| \xi-\xi^{\star} \|_p$, and thus
the variability of the $E_i^j(\xi)$s will be small if 1) the reference point $\xi^{\star}$ is close to the mode of the posterior and 2) $\xi$ is close to $\xi^{\star}$. Under standard asymptotics we expect a draw 
from the posterior for $\xi$ to be $O_p(n^{-1/2})$ from the posterior mode. Thus if we have a procedure for finding a reference point $\xi^{\star}$ which is within $O(n^{-1/2})$ of the posterior mode then this would
ensure $\|\xi-\xi^{\star}\|_2$ is $O_p(n^{-1/2})$ if $\xi$ is drawn from the posterior. For such a choice of $\xi^{\star}$ we would have $\partial_i\Psi(\xi^{\star})$ of $O_p(n^{1/2})$.

Using the Lipschitz condition, we can now obtain computational bounds of $(m_i)$ for a trajectory $\xi(t) := \xi + \theta t$ originating in $(\xi, \theta)$. Define 
\[ M_i(t) := a_i + b_i t, \quad t \geq 0, \quad i = 1, \dots, d,  \]
where $a_i := \left( \theta_i \partial_i\Psi(\xi^{\star})\right)^+ + C_i \| \xi - \xi^{\star}\|_p$ and $b_i := C_i d^{1/p}$. Then~\eqref{eq:comp-bound-sub-sampling} is satisfied. Indeed, using Lipschitz continuity of $y \mapsto (y)^+$,
\begin{align*}
 m_i^j(t) & = \left( \theta_i E_i^j(\xi + \theta t) \right)^+ = \left( \theta_i\partial_i\Psi(\xi^{\star})+\theta_i \partial_i \Psi^j(\xi + \theta t) - \theta_i \partial_i \Psi^j(\xi^{\star}) \right)^+ \\
 & \leq \left( \theta_i \partial_i\Psi(\xi^{\star})\right)^+ +\left|\partial_i \Psi^j(\xi) -\partial_i \Psi^j(\xi^{\star}) \right| + \left| \partial_i \Psi^j(\xi + \theta t) - \partial_i \Psi^j(\xi) \right| \\
 & \leq \left( \theta_i \partial_i\Psi(\xi^{\star})\right)^+ + C_i \left( \|\xi - \xi^{\star}\|_p + t \| \theta\|_p \right) = M_i(t).
%
\end{align*}

Implementing this scheme requires some pre-processing of the data. First we need a way of choosing a suitable reference point $\xi^{\star}$ to find a value close 
to the mode using an approximate or exact numerical optimization routine. The complexity of this operation will be $O(n)$.
Once we have found such a reference point we have an one-off $O(n)$ cost of calculating $\partial_i\Psi(\xi^{\star})$ for each $i=1,\ldots,d$. 
However, once we have paid this upfront computational cost, the resulting Zig-Zag sampler can be super-efficient. This is discussed in more detail in Section~\ref{sec:scaling}, and 
demonstrated empirically in Section~\ref{sec:examples}.
The version of Algorithm~\ref{alg:sub-sampling} resulting from this choice of $E^j_i$ and $M_i$ will be called Zig-Zag with Control Variates (ZZ-CV).



\begin{remark} 
When choosing $p \geq 1$, there will be a trade-off between the magnitude of $C_i$ and of $\| \xi -\xi^{\star}\|_p$, which may influence the scaling of Zig-Zag sampling with dimension. We will see in Section~\ref{sec:gaussian} that for i.i.d. Gaussian components, the choice $p = \infty$ is optimal. When the situation is less clear, choosing the Euclidean norm ($p = 2$) is a reasonable choice.
\end{remark}

\section{Scaling analysis}
\label{sec:scaling}

In this section we provide an informal scaling argument for canonical Zig-Zag, and Zig-Zag with control variates and sub-sampling.
For the moment fix $n \in \N$ and consider a posterior with negative log density
\[
 \Psi(\xi)=-\sum_{j=1}^n \log f(x^j \mid \xi),
\]
where $x^j$ are i.i.d. drawn from $f(x^j \mid \xi_0)$. 
Let $\widehat{\xi}$ denote the maximum likelihood estimator (MLE) for $\xi$ based on data $x^1,\ldots,x^n$. Introduce the coordinate transformation
\[
 \phi(\xi) =\sqrt{n}(\xi-\widehat{\xi}), \quad  \xi(\phi) = \frac 1{\sqrt n} \phi + \widehat{\xi}.
\]

As $n\rightarrow\infty$ the posterior distribution in terms of $\phi$ will converge to a multivariate Gaussian distribution with mean 0 and covariance matrix given by the inverse of the expected information $i(\theta_0)$; see e.g. \cite{Johnson1970}.


\subsection{Scaling of Zig-Zag Sampling (ZZ)}
\label{sec:scaling-ZZ}

First let us obtain a Taylor expansion of the switching rate for $\xi$ close to $\widehat \xi$. We have
\begin{align*}
 & \partial_{\xi_i} \Psi(\xi) = - \partial_{\xi_i} \sum_{j=1}^n \log f(x^j \mid \xi) \\
 & = \underbrace{- \partial_{\xi_i} \sum_{j=1}^n \log f(x^j \mid \widehat \xi)}_{= 0} - \sum_{j=1}^n  \sum_{k=1}^d \partial_{\xi_i}\partial_{\xi_k} \log f(x^j \mid \widehat \xi) (\xi_k - \widehat \xi_k) +  O( \|\xi - \widehat \xi\|^2).
\end{align*}
The first term vanishes by the definition of the MLE.
Expressed in terms of $\phi$, the switching rates are
\begin{align*} (\theta_i \partial_{\xi_i} \Psi(\xi(\phi)))^+
& =\underbrace{\frac 1 {\sqrt{n}} \left( - \sum_{j=1}^n \sum_{k=1}^d \partial_{\xi_i} \partial_{\xi_k} \log f(x^j \mid \widehat \xi) \phi_k \right)^+}_{O(\sqrt{n})} + O\left( \frac{\| \phi\|^2}{n} \right).
%
\end{align*}
With respect to the coordinate $\phi$, the canonical Zig-Zag process has constant speed $\sqrt{n}$ in each coordinate, and by the above computation, a switching rate of $O(\sqrt{n})$. After a rescaling of the time parameter by a factor $\sqrt{n}$, the process in the $\phi$-coordinate becomes a Zig-Zag process with unit speed in every direction and switching rates
\[ \left( - \frac 1 n \sum_{j=1}^n \sum_{k=1}^d \partial_{\xi_i} \partial_{\xi_k} \log f(x^j \mid \xi) \phi_k \right)^+ + O(n^{-1/2}).\]
If we let $n \rightarrow \infty$, the switching rates converge almost surely to those of a Zig-Zag process with switching rates
\[ \widetilde \lambda_i(\phi, \theta) = \left(\theta_i (i(\theta_0) \phi)_i \right)^+ \]
where $i(\theta_0)$ denotes the expected information. These switching rates correspond to the limiting Gaussian distribution with covariance matrix $(i(\theta_0))^{-1}$.

In this limiting Zig-Zag process, all dependence on $n$ has vanished. Starting from equilibrium, we require a time interval of $O(1)$ (in the rescaled time) to obtain an essentially independent sample. In the original time scale this corresponds to a time interval of $O(n^{-1/2})$. As long as the computational bound in the Zig-Zag algorithm is  $O(n^{1/2})$, this can be achieved using $O(1)$ proposed switches. The computational cost for every proposed switch is $O(n)$, because the full data $(x^i)_{i=1}^n$ needs to be processed in the computation of the true switching rate at the proposed switching time.

\emph{We conclude that the computational complexity of the Zig-Zag (ZZ) algorithm per independent sample is $O(n)$, provided that the computational bound is $O(n^{1/2})$.}
This is the best we can expect for any standard Monte Carlo algorithm (where we will have a $O(1)$ number of iterations, but each iteration is $O(n)$ in computational cost).

To compare, if the computational bound is $O(n^{\alpha})$ for some $\alpha > 1/2$, then we require $O(n^{\alpha - 1/2})$ proposed switches before we have simulated a total time interval of length $O(n^{-1/2})$, so that, with a complexity of $O(n)$ per proposed switching time, the Zig-Zag algorithm has total computational complexity $O(n^{\alpha + 1/2})$.
So, for example, with global bounds we have that the computational bound is $O(n)$ (as each term in the log density is $O(1)$), and hence ZZ will have total computational complexity of $O(n^{3/2})$.

\begin{example}[Dominated Hessian]
Consider Algorithm~\ref{alg:dominated-hessian} in the one-dimensional case, with the second derivative of $\Psi$ bounded from above by $Q > 0$. We have $Q = O(n)$ as $\Psi''$ is the sum of $n$ terms of $O(1)$. The value of $b$ is kept fixed at the value $b = Q = O(n)$. Next $a$ is given initially as 
\begin{align*} a = \theta \Psi'(\xi) \leq \theta \underbrace{\Psi'(\widehat \xi)}_{= 0} + \underbrace{Q}_{O(n)} \underbrace{(\xi - \widehat \xi)}_{O(n^{-1/2})} = O(n^{1/2}),
\end{align*}
and increased by $b \tau$ until a switch happens and $a$ is reset to $\theta \Psi'(\xi)$. Because of the initial value for $a$, switches will occur at rate $O(n^{1/2})$ so that $\tau$ will be $O(n^{-1/2})$, and the value of $a$ will remain $O(n^{1/2})$. Hence the magnitude of the computational bound $M(t) = (a + b t)^+$ is $O(n^{1/2})$.
\end{example}

\subsection{Scaling of Zig-Zag with Control Variates (ZZ-CV)}
\label{sec:scaling-ZZCV}

Now we will study the limiting behaviour as $n \rightarrow \infty$ of 
ZZ-CV introduced in Section~\ref{sec:sub-sampling-control-variates}. In determining the computational bounds we take $p = 2$ for simplicity, e.g. in~\eqref{eq:sub-sampling-lipschitz-bound}.
Also for simplicity assume that $\xi \mapsto \partial_{\xi_i} \log f(x^j \mid \xi)$ has Lipschitz constant $k_i$ (independent of $j = 1, \dots, n$) and write $C_i = n k_i$, so that~\eqref{eq:sub-sampling-lipschitz-bound} is satisfied. 
In practice there may be a logarithmic increase with $n$ in the Lipschitz constants $k_i$ as we have to take a global bound in $n$. For the present discussion we ignore such logarithmic factors.
We assume reference points $\xi^{\star}$ for growing $n$ are determined in  
such a way that $\|\xi^{\star}-\widehat{\xi}\|_2$ is $O(n^{-1/2})$. For definiteness, suppose there exists a $d$-dimensional random variable $Z$ such that $n^{1/2} (\xi^{\star} - \widehat \xi) \rightarrow Z$ in distribution, with the randomness in $Z$ independent of $(x^j)_{j=1}^{\infty}$.

We can look at ZZ-CV with respect to the scaled coordinate $\phi$ as $n\rightarrow \infty$.  Denote the reference point for the rescaled parameter as $\phi^{\star} :=\sqrt{n}(\xi^{\star}-\widehat{\xi})$.

The essential quantities to consider are the switching rate estimators $E_i^j$. We estimate
\begin{align*}
|E_i^j(\xi)| & = \left| \partial_{\xi_i}\Psi(\xi^{\star})+ \partial_{\xi_i} \Psi^j(\xi) - \partial_{\xi_i} \Psi^j(\xi^{\star}) \right|  \\
& = \left| \partial_{\xi_i}\Psi(\xi^{\star}) - \partial_{\xi_i} \Psi(\widehat \xi) + \partial_{\xi_i} \Psi^j(\xi) - \partial_{\xi_i} \Psi^j(\xi^{\star}) \right| \\
& \leq \underbrace{C_i}_{O(n)} \underbrace{\|  \xi^{\star} - \widehat \xi \|}_{O(n^{-1/2})} + \underbrace{C_i}_{O(n)} \underbrace{\| \xi - \xi^{\star} \|}_{O(n^{-1/2})}.
\end{align*}
We find that $|E_i^j(\xi)| = O(n^{1/2})$ under the stationary distribution.

By slowing down the Zig-Zag process in $\phi$ space by $\sqrt{n}$, the continuous time process generated by ZZ-CV will approach a limiting Zig-Zag process with a certain switching rate of $O(1)$. In general this switching rate will depend on the way that $\xi^{\star}$ is obtained. To simplify the exposition, in the following computation we assume $\xi^{\star} = \widehat \xi$.
Rescaling by $n^{-1/2}$, and developing a Taylor approximation around $\widehat \xi$, 
\begin{align*}
 n^{-1/2} E_i^j(\xi) & = n^{-1/2} \left(\partial_{\xi_i} \Psi^j(\xi) - \partial_{\xi_i} \Psi^j(\widehat \xi)  \right) \\
 & = n^{-1/2} \left( - n \partial_{\xi_i} \log f(x^j \mid \xi) + n \partial_{\xi_i} \log f(x^j \mid \widehat \xi) \right) \\
 & = -n^{1/2} \left( \sum_{k=1}^d  \partial_{\xi_i} \partial_{\xi_k} \log f(x^j \mid \widehat \xi) (\xi_k - \widehat \xi_k) \right) + O(n^{1/2}\| \xi - \widehat \xi \|^2) \\
 & = - \sum_{k=1}^d \partial_{\xi_i} \partial_{\xi_k} \log f(x^j \mid \widehat \xi)  \phi_k + O(n^{-1/2}).
 \end{align*}
By Theorem~\ref{thm:sub-sampling-rate}, the rescaled effective switching rate for ZZ-CV is given by
\begin{align*}
 \widetilde \lambda_i(\phi, \theta) & := n^{-1/2} \lambda_i(\xi(\phi),\theta) = \frac 1 {n^{3/2}} \sum_{j=1}^n \left( \theta_i E_i^j (\xi(\phi)) \right)^+ \\
 & = \frac 1 n \sum_{j=1}^n \left(- \theta_i \sum_{k=1}^d \partial_{\xi_i} \partial_{\xi_k} \log f(x^j \mid \widehat \xi) \phi_k \right)^+ + O(n^{-1/2}) \\
 & \rightarrow \E \left( - \theta_i \sum_{k=1}^d \partial_{\xi_i} \partial_{\xi_k} \log f(X \mid \xi_0 ) \phi_k \right)^+,
\end{align*}
where $\E$ denotes expectation with respect to $X$, with density $f(\cdot \mid \xi_0)$, and the convergence is a consequence of the law of large numbers. If $\xi^{\star}$ is not exactly equal to $\widehat \xi$, the limiting form of $\widetilde \lambda_i(\phi,\theta)$ will be different, but the important point is that it will be $O(1)$, which follows from the bound on $|E^j_i|$ above.

Just as with ZZ, the rescaled Zig-Zag process underlying ZZ-CV converges to a limiting Zig-Zag process with switching rate $\widetilde \lambda_i(\phi, \theta)$.
Since the computational bounds of ZZ-CV are $O(n^{1/2})$, a completely analogous reasoning to the one for ZZ algorithm above (Section~\ref{sec:scaling-ZZ}) leads to the conclusion that $O(1)$ proposed switches are required to obtain an independent sample. However, in contrast with the ZZ-algorithm, the ZZ-CV algorithm is designed in such a way that the computational cost per proposed switch is $O(1)$.

\emph{We conclude that the computational complexity of the ZZ-CV algorithm is $O(1)$ per independent sample.} This provides a factor $n$ increase in efficiency over standard MCMC algorithms, resulting in an \emph{asymptotically unbiased} algorithm for which \emph{the computational cost of obtaining an independent sample does not depend on the size of the data}.

\subsection{Remarks}


The arguments above assume we are at stationarity -- and how quickly the two algorithms converge is not immediately clear. 
Note however that for sub-sampling Zig-Zag it is possible to choose the reference point $\xi^{\star}$ as starting point, thus avoiding much of the issues about convergence.



In some sense, the good computational scaling of ZZ-CV is leveraging the asymptotic normality of the posterior, but in such a way that ZZ-CV always samples from the true posterior. 
Thus when the posterior is close to Gaussian it will be quick; when it is far from Gaussian it may well be slower but will still be ``correct''. 
This is fundamentally different from other algorithms \cite[e.g.][]{Neiswanger:2013,Scott:2013,Bardenet2015} that utilise the asymptotic normality in terms of justifying their approximation to the posterior.
Such algorithms are accurate if the posterior is close to Gaussian, but may be inaccurate otherwise, and it is often impossible to quantify the size of the approximation in practice.



\section{Examples and experiments}
\label{sec:examples}

\subsection{Sampling and integration along Zig-Zag trajectories}


There are essentially two different ways of using the Zig-Zag skeleton points which we obtain by using e.g. Algorithms~\ref{alg:general}, \ref{alg:dominated-hessian}, or \ref{alg:sub-sampling}.

The first possible approach is to collect a number of samples along the trajectories. Suppose we have simulated the Zig-Zag process up to time $\tau > 0$, and we wish to collect $m$ samples. 
This can be achieved by setting $t_i = i \tau/m$, and setting $\Xi_i := \Xi(t_i)$ for $i = 1, \dots, m$, with the continuous time trajectory $(\Xi(t))$ defined as in Section~\ref{sec:construction}. 
In order to approximate $\pi(f)$ numerically for some function $f : \R^d \rightarrow \R$ of interest, we can use the usual ergodic average
\[ \widehat{\pi(f)} :=  \frac 1 m  \sum_{i=1}^m f(\Xi_i).\]
We can also estimate posterior quantiles by using the quantiles of the sample $\Xi_1,\ldots,\Xi_m$, as with standard MCMC output.
An issue with this approach is that we have to decide on the number, $m$, of samples we wish to use. Whilst the more samples we use the greater the accuracy of our approximation
to $\pi(f)$, this comes at an increased computational and storage cost. The trade-off in choosing an appropriate value for $m$ is equivalent to the choice of how much to thin output from a standard MCMC algorithm.


It is important that one does not make the mistake of using the switching points of the Zig-Zag process as samples, as these points are not distributed according to $\pi$. 
In particular, the switching points are biased towards the tails of the target distribution.

An alternative approach is intrinsically related to the continuous time and piecewise linear nature of the Zig-Zag trajectories. This approach consists of continuous 
time integration of the Zig-Zag process.  
By the continuous time ergodic theorem, for $f$ as above, $\pi(f)$ can be estimated as
\[ \widehat {\pi(f)} = \frac 1 {\tau} \int_0^{\tau} f(\Xi(s)) \ d s.\]
Since the output of the Zig-Zag algorithms consists of a finite number of skeleton points $(T^i, \Xi^i, \Theta^i)_{i=0}^k$, we can express this as
\[ \widehat {\pi(f)} =  \frac 1 {T^k} \sum_{i=1}^k \int_{T^{i-1}}^{T^i} f(\Xi^{i-1} + \Theta^{i-1} (s-T^{i-1})) \ d s.\]
Due to the piecewise linearity of $\Xi(t)$, in many cases these integrals can be computed exactly, e.g. for the moments, $f(x) = x^p$, $p \in \R$.
In cases where the integral can not be computed exactly, numerical quadrature rules can be applied.
An advantage of this method is that we do not have to make an arbitrary decision on the number of samples to extract from the trajectory.

\subsection{Beating one ESS per epoch}
\label{sec:ess-per-epoch}

We use the term \emph{epoch} as a unit of computational cost, corresponding to the number of iterations required to evaluate the complete gradient of $\log \pi$. This means that for the basic Zig-Zag algorithm (without sub-sampling), an epoch consists of exactly one iteration, and for the sub-sampled variants of the Zig-Zag algorithm, an epoch consists of $n$ iterations. The CPU running times per epoch of the various algorithms we consider are equal up to a constant factor. To assess the scaling of various algorithms, we use \emph{ESS per epoch}. The notion of ESS is discussed in the supplementary material \cite[Section 2]{supplement}. Consider any classical MCMC algorithm based upon the Metropolis-Hastings acceptance rule. Since every iteration requires an evaluation of the full density function to compute the acceptance probability, we have that the ESS per epoch for such an algorithm is bounded from above by one. Similar observations apply to all other known MCMC algorithms capable of sampling asymptotically from the exact target distribution.

There do exist several conceptual innovations based on the idea of sub-sampling, which have some theoretical potential to overcome the fundamental limitation of one ESS per epoch sketched above. 


The Pseudo-Marginal Method (PMM, \cite{AndrieuRoberts2009}) is based upon using a positive unbiased estimator for a possibly unnormalized density. Obtaining an unbiased estimator of a product is much more difficult than obtaining one for a sum. Furthermore,
it has been shown to be impossible to construct an estimator that is guaranteed to be positive without other information about the product, such as a bound on the terms in the product (\cite{Jacob:2015}).
Therefore the PMM does not apply in a straightforward way to vanilla MCMC in Bayesian inference.

In the supplementary material \cite[Section 3]{supplement} we analyse the scaling of Stochastic Gradient Langevin Dynamics (SGLD, \cite{WellingTeh2011}) in an analogous fashion to the analysis of ZZ and ZZ-CV in Section~\ref{sec:scaling}. From this analysis we conclude that it is in general not possible to implement SGLD in such a way that the ESSpE has a larger order of magnitude than $O(1)$. We compare SGLD to Zig-Zag in experiments of Sections~\ref{sec:gaussian} and~\ref{sec:nonidentifiable}.

\subsection{Mean of a Gaussian distribution}
\label{sec:gaussian}

Consider the illustrative problem of estimating the mean of a Gaussian distribution. This problem has the advantage that it allows for an analytical solution which can be compared with the numerical solutions obtained by Zig-Zag Sampling and other methods. Conditional on a one-dimensional parameter $\xi$, the data is assumed to be i.i.d. from $N(\xi, \sigma^2)$. Furthermore a $N(0, 1/\rho^2)$ prior on $\xi$ is specified. Data are generated from the true distribution $N(\xi_0, \sigma^2)$ for some fixed $\xi_0$.
For a detailed description of the experiment and computational bounds, see Section 4 of the supplementary material.


In this experiment, we compare the mean square error (MSE) for several algorithms, namely basic Zig-Zag (ZZ), Zig-Zag with Control Variates (ZZ-CV), Zig-Zag with 
Control Variates  with a ``sub-optimal'' reference point (ZZ-soCV), and Stochastic Gradient Langevin Dynamics (SGLD). SGLD is implemented with fixed step size, as is usually done in practice, see e.g. \cite{Vollmer2015}, with the added benefit that it makes the comparison with the Zig-Zag algorithms more straightforward. Here in basic Zig-Zag we pretend that every iteration requires the evaluation of $n$ observations (whereas in practice, we can pre-compute $\xi^{\mathrm{MAP}}$). 
 
Results for this experiment are displayed in Figure~\ref{fig:gaussian}. The MSE for the second moment using SGLD does not decrease beyond a fixed value, indicating the presence of bias in SGLD. This bias does not appear in the different versions of Zig-Zag sampling, 
agreeing with the theoretical result that ergodic averages over Zig-Zag trajectories are consistent. 
Furthermore we see a significant relative increase in efficiency for ZZ-(so)CV over basic ZZ when the number of observations is increased, agreeing with the scaling results of Section~\ref{sec:scaling}. A poor choice of reference point (as in ZZ-soCV) is seen to have only a small effect on the efficiency.


\begin{figure}
 \begin{subfigure}[b]{0.45\textwidth}
  \includegraphics[width=\textwidth]{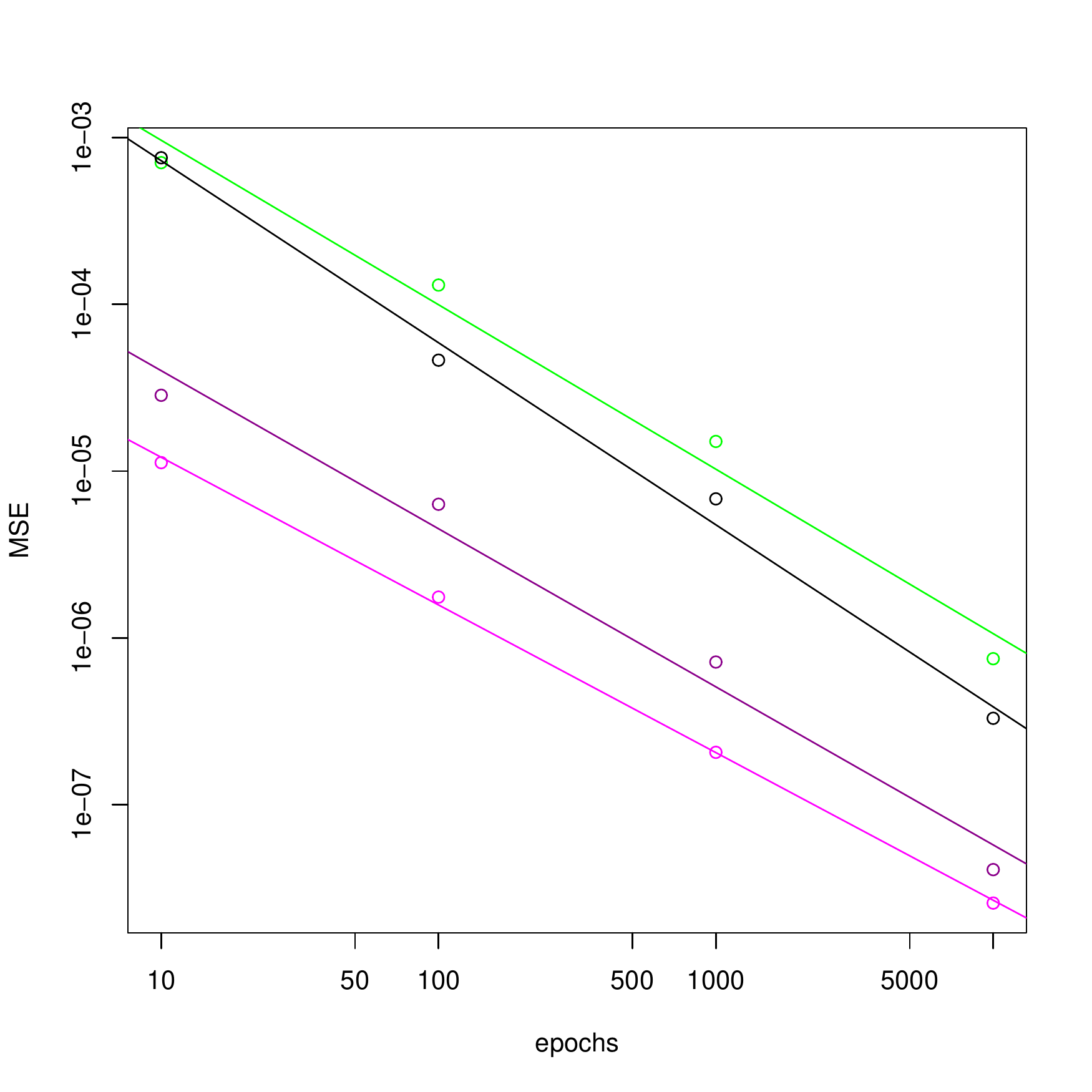}
  \caption{First moment, 100 observations}
 \end{subfigure}
  \begin{subfigure}[b]{0.45\textwidth}
  \includegraphics[width=\textwidth]{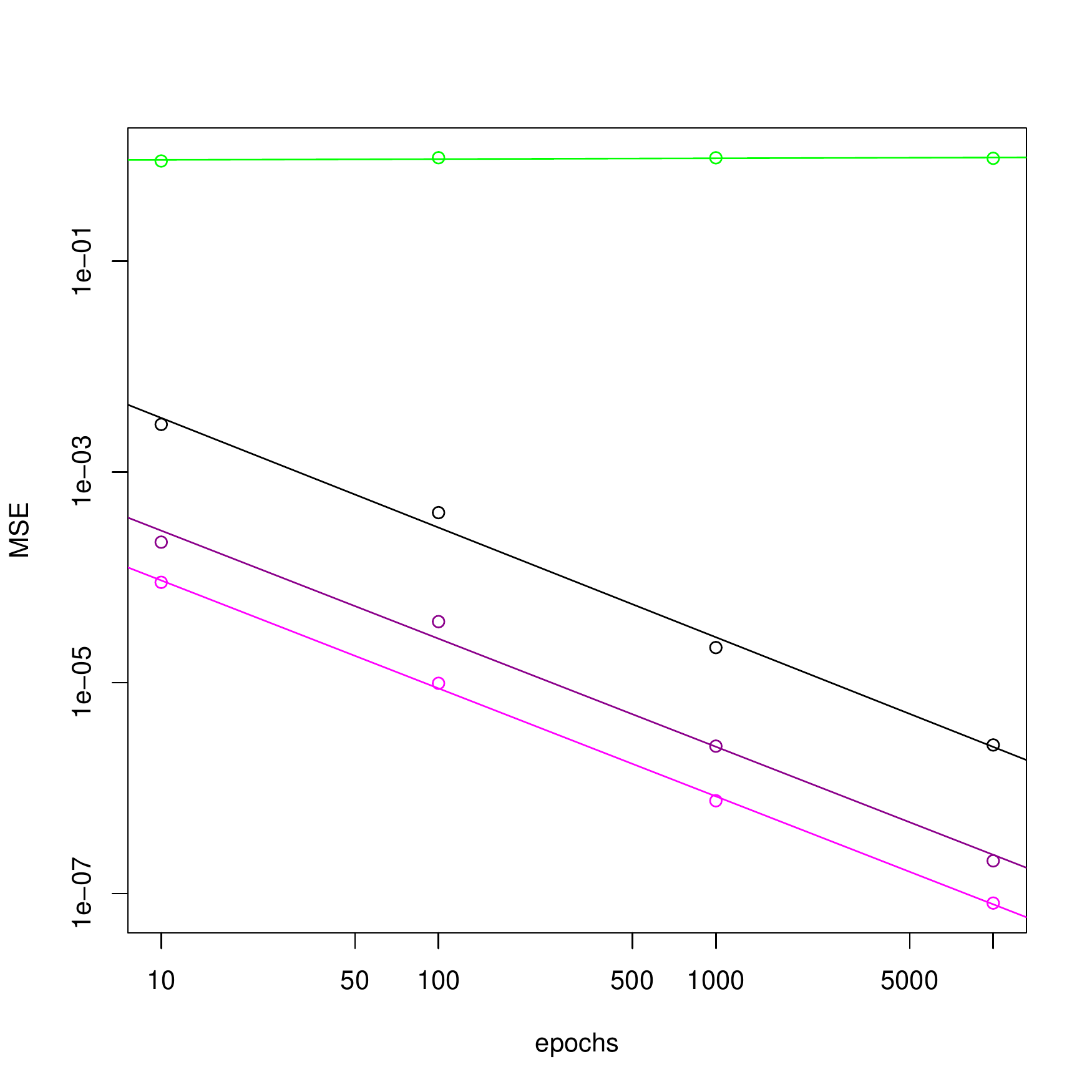}
  \caption{Second moment, 100 observations}
 \end{subfigure} \\
  \begin{subfigure}[b]{0.45\textwidth}
  \includegraphics[width=\textwidth]{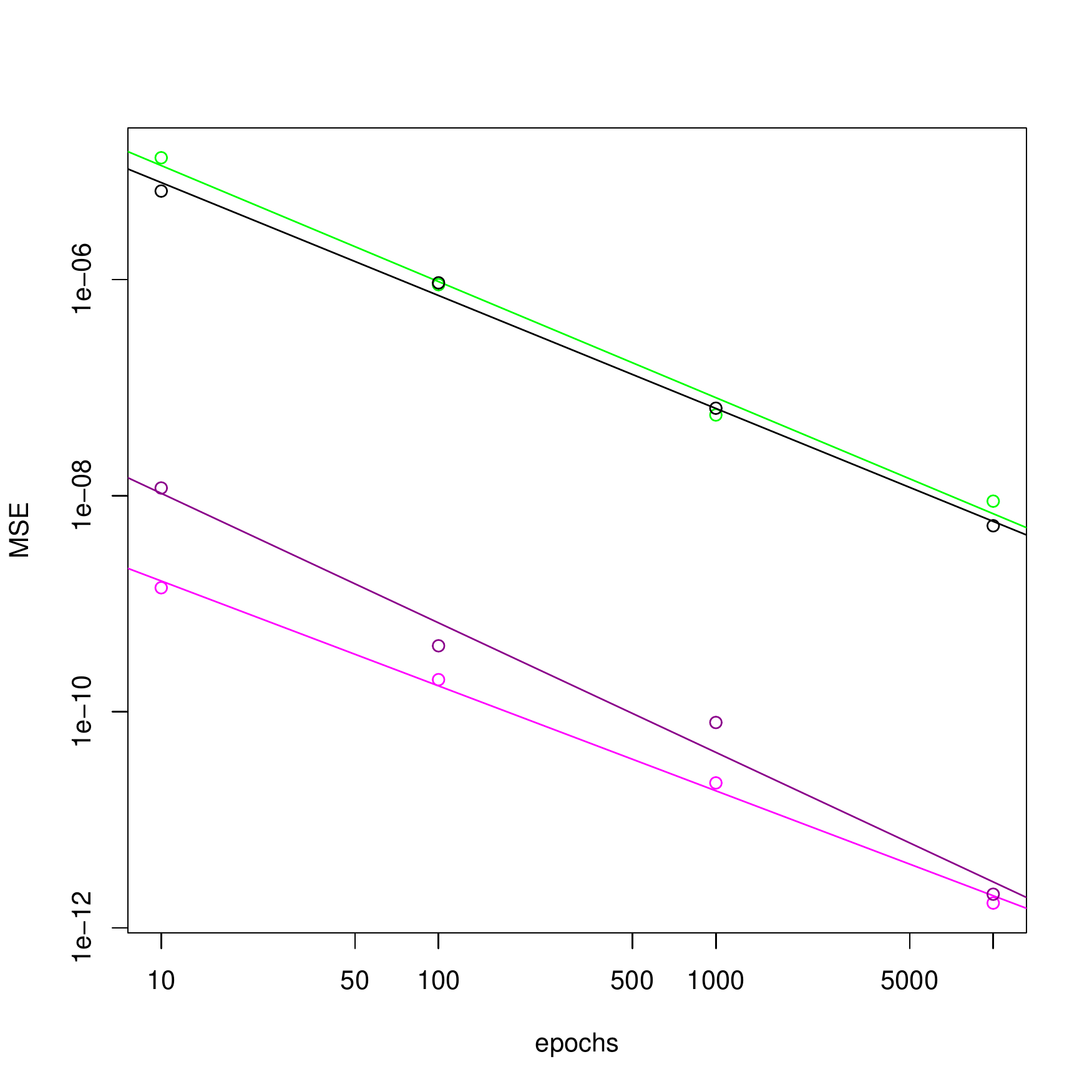}
  \caption{First moment, $10^4$ observations}
 \end{subfigure}
  \begin{subfigure}[b]{0.45\textwidth}
  \includegraphics[width=\textwidth]{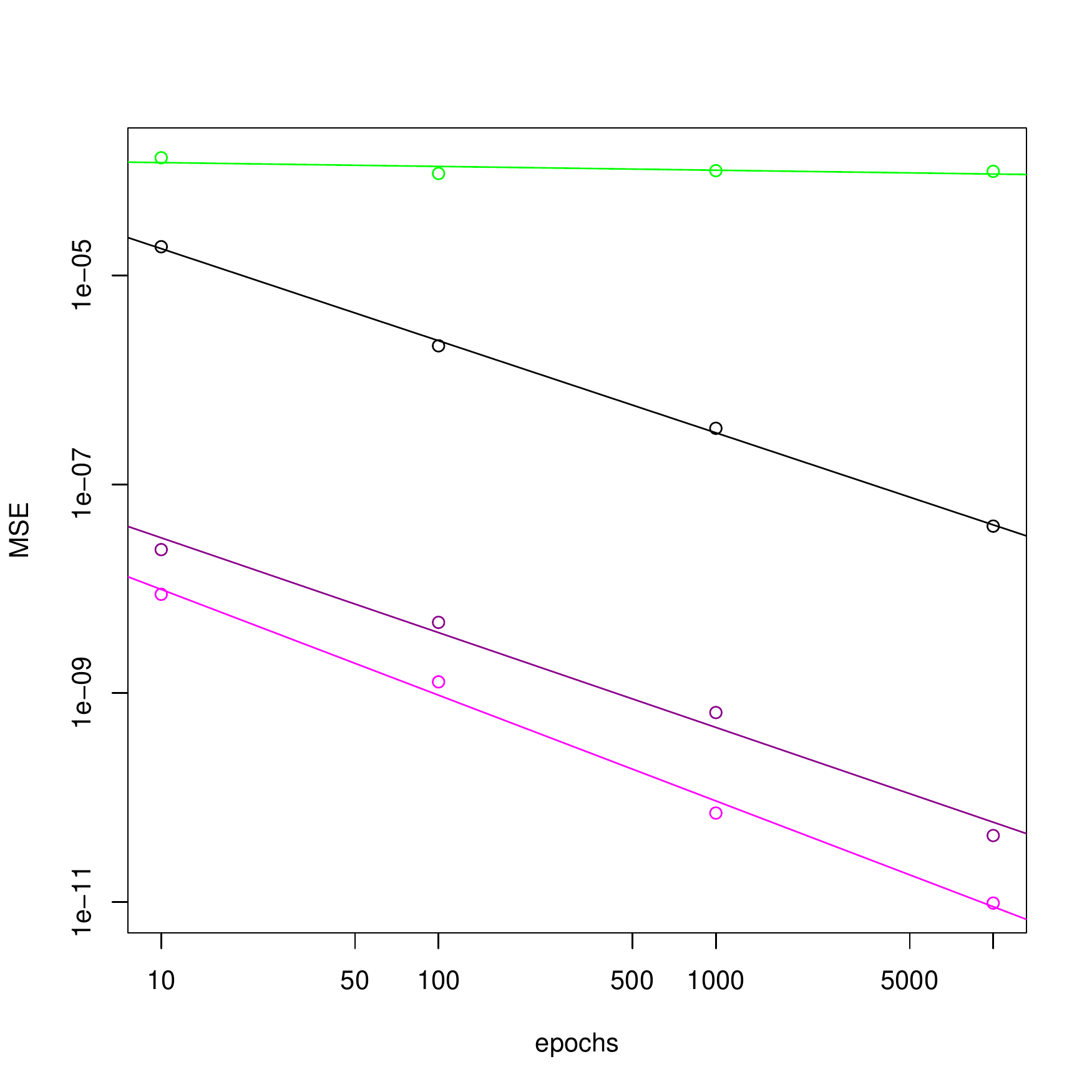}
  \caption{Second moment, $10^4$ observations}
  \end{subfigure}
 
\caption[LoF variant]{Mean square error (MSE) in the first and second moment as a function of the number of epochs, based on $n = 100$  or $n = 10,000$ observations, 
 for a one-dimensional Gaussian posterior distribution (Section~\ref{sec:gaussian}). Displayed are SGLD (green), ZZ-CV (magenta), ZZ-soCV (dark magenta), ZZ (black). 
  The displayed dots represent averages over experiments based on randomly generated data from the true posterior distribution. 
    Parameter values (see \cite[Section 4]{supplement}) are $\xi_0 = 1$ (the true value of the mean parameter), $\rho = 1$, $\sigma = 1$ and $c_1 = 1$, $c_2 = 1/100$ (for the SGLD parameters, see the supplement, \cite[Section 3]{supplement}). 
The value of $\xi^{\star}$ for ZZ-soCV is based on a sub-sample of size $m = n/10$ so that it will not be equal to the exact maximizer of the posterior. For an honest comparison, trajectories of all algorithms have initial condition equal to $\xi^{\mathrm{MAP}}$.}
 \label{fig:gaussian}
\end{figure}

\subsection{Logistic regression}
\label{sec:logistic}



In this numerical experiment we compare how the ESS per epoch (ESSpE) and ESS per second grow with the number of observations $n$ for several Zig-Zag algorithms and the MALA algorithm when applied to a logistic regression problem. Conditional on a $d$-dimensional parameter $\xi$ and given $d$-dimensional covariates $x^j \in \R^d$, where $j = 1, \dots, n$, and with $x_1^j = 1$ for all $j$, the binary variable $y^j \in \{0,1\}$ has distribution
\[ \P(y^j \mid x_1^j, \dots, x_d^j, \xi_1, \dots, \xi_d) = \frac 1 { 1 + \exp \left( - \sum_{i=1}^d \xi_i x_i \right)}.\]
Combined with a flat prior distribution, this induces a posterior distribution $\xi$ given observations of $(x^j, y^j)$ for $j =1, \dots, n$; see the supplementary material for implementational details \cite[Section 5]{supplement}.

The results of this experiment are shown in Figure~\ref{fig:ess-per-epoch}. In both the plots of ESS per epoch (see (a) and (c)), the best linear fit for ZZ-CV has slope approximately 0.95, 
which is in close agreement with the scaling analysis of Section~\ref{sec:scaling}. 
The other algorithms have roughly a horizontal slope, corresponding to a linear scaling with the size of the data. We conclude that, among the algorithms tested, ZZ-CV is the only algorithm for which the ESS per CPU second is approximately constant
as a function of the size of the data (see Figure~\ref{fig:ess-per-epoch}, (b) and (d)). Furthermore ZZ-CV obtains an ESSpE which is roughly linearly increasing with the number of observations $n$ (see Figure~\ref{fig:ess-per-epoch},(a) and (c)). 
whereas the other versions of the Zig-Zag algorithms, and MALA, have an ESSpE which is approximately constant with respect to $n$. 
These statements apply regardless of the dimensionality of the problem.

\begin{figure}
\begin{subfigure}[b]{0.45 \textwidth}
 \includegraphics[width = \textwidth]{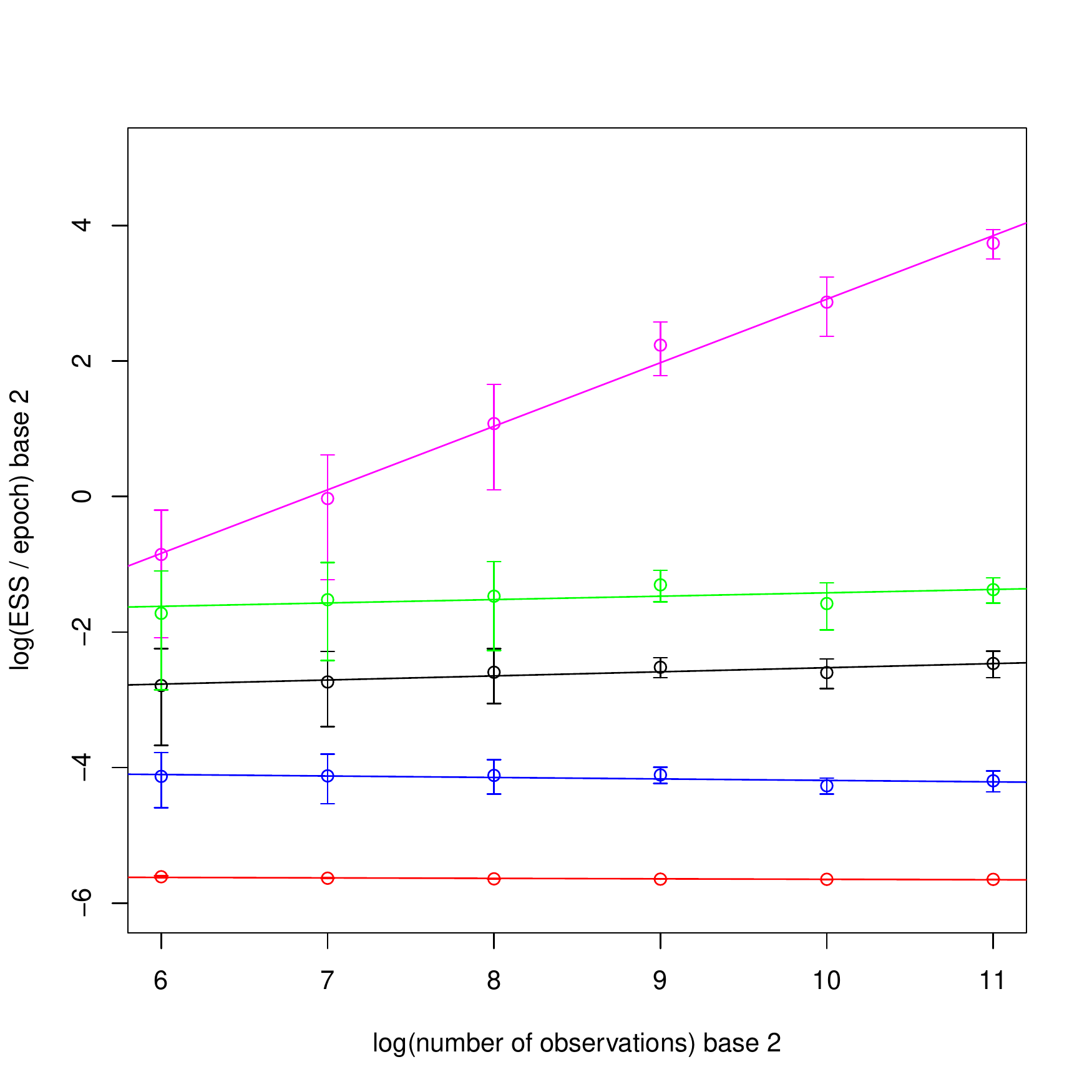}
 \caption{ESS per epoch, 2 dimensions}
\end{subfigure}
\begin{subfigure}[b]{0.45 \textwidth}
 \includegraphics[width = \textwidth]{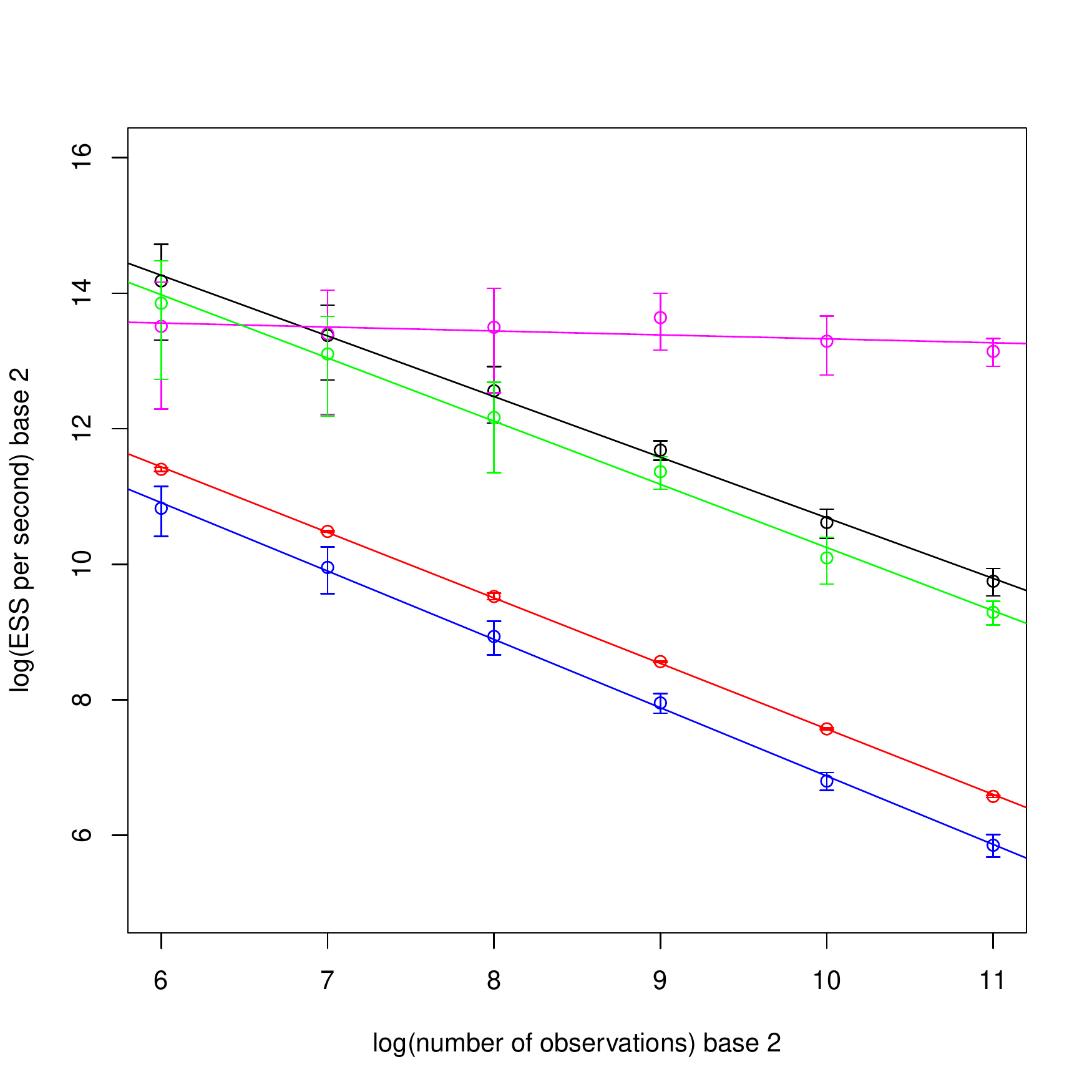}
 \caption{ESS per second, 2 dimensions}
\end{subfigure}
\begin{subfigure}[b]{0.45 \textwidth}
 \includegraphics[width = \textwidth]{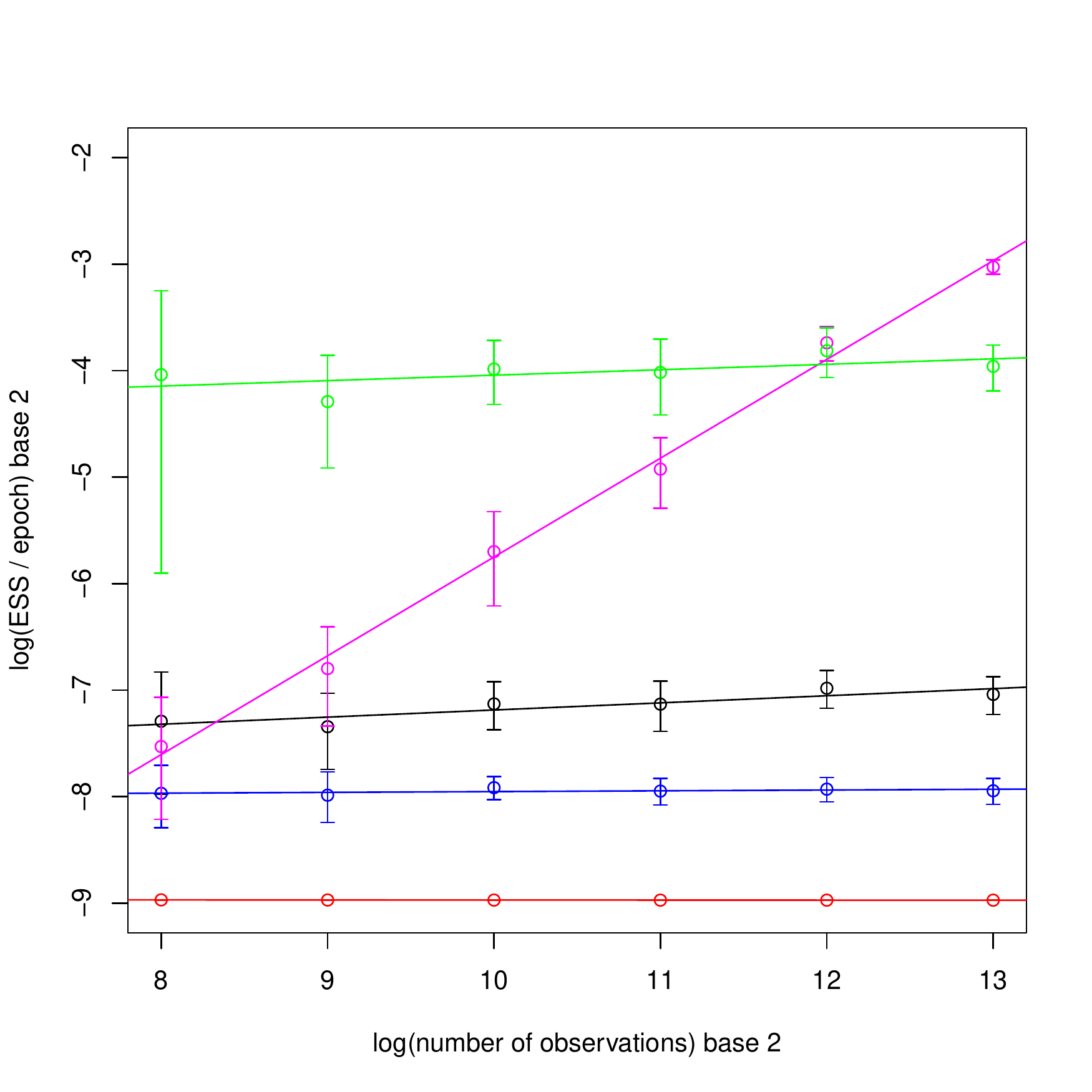}
 \caption{ESS per epoch, 16 dimensions}
\end{subfigure}
\begin{subfigure}[b]{0.45 \textwidth}
 \includegraphics[width = \textwidth]{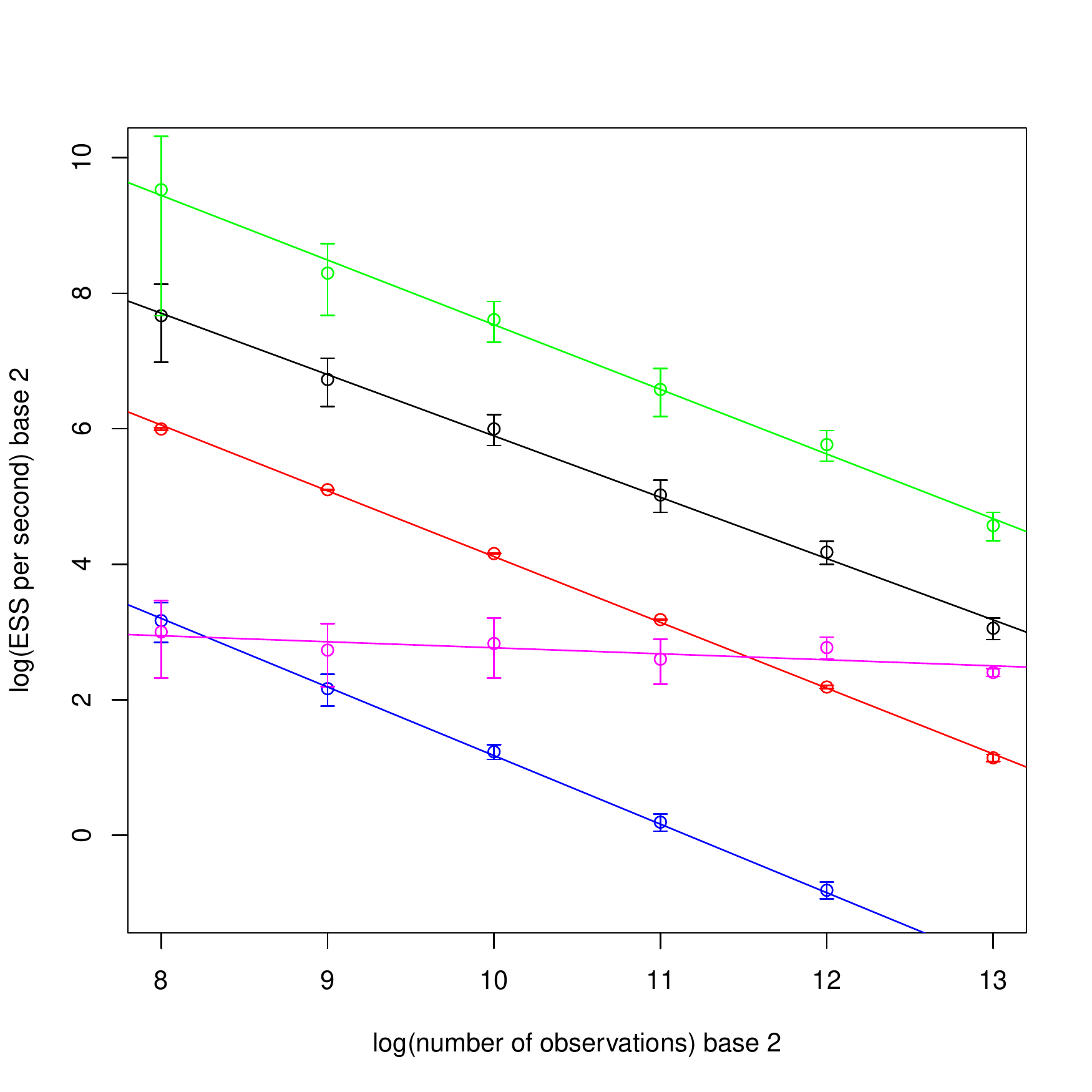}
 \caption{ESS per second, 16 dimensions}
\end{subfigure}
\caption{Log-log plots of the experimentally observed dependence of ESS per epoch (ESSpE) and ESS per second (ESSpS) with respect to the first coordinate $\Xi^1$, as a function of the number of observations $n$ in the case of (2-D and 16-D) Bayesian logistic regression (Section~\ref{sec:logistic}). 
Data is randomly generated based on true parameter values $\xi_0 = (1,2)$ (2-D) and $\xi_0 = (1, \dots, 1)$ (16-D). 
Trajectories all start in the true parameter value $\xi_0$. Plotted are mean and standard deviation over 10 experiments,
along with the best linear fit. Displayed are MALA (tuned to have optimal acceptance ratio, green), Zig-Zag with global bound (red), Zig-Zag with Lipschitz bound (black), ZZ-SS using global bound (blue) and ZZ-CV (magenta), all run for $10^5$ epochs. As reference point for ZZ-CV we compute the posterior mode numerically, the cost of which is negligible compared to the MCMC. The experiments are carried out in R with C++ implementations of all algorithms.
\label{fig:ess-per-epoch}}
\end{figure}

\subsection{A non-identifiable logistic regression example with unbounded Hessian}
\label{sec:nonidentifiable}

In a further experiment we consider one-dimensional data $(x^j, y^j)$, for $j = 1, \dots, n$, $x^j \in \R$, $y^j \in \{0,1\}$, which we assume for illustrational purposes to be generated from a logistic model where $\P(y^j = +1 \mid x^j, \xi_1, \xi_2) = \frac {1}{1 + \exp(-(\xi_1 + \xi_2^2) x^j)}$. The model is non-identifiable since two parameters $\xi$, $\eta$ correspond to the same model as long as $\xi_1 + \xi_2^2 = \eta_1 + \eta_2^2$. This leads to a sharply rigged probability density function  reminiscent of density functions concentrated along lower dimensional submanifolds which often arise in Bayesian inference problems. In this case the Hessian of the log density is unbounded so that we cannot use the standard framework for the Zig-Zag algorithms.
It is discussed in the supplementary material \cite[Section 6]{supplement}, how to obtain computational bounds for the Zig-Zag and ZZ-CV algorithms, which may serve as an illustration on how to obtain such bounds in settings beyond those described in Sections~\ref{sec:sampling-dominated-hessian} and~\ref{sec:sub-sampling-control-variates}. 

In Figure~\ref{fig:nonidentifiable} we compare trace plots for the Zig-Zag algorithms (ZZ, ZZ-CV) 
to trace plots for Stochastic Gradient Langevin Dynamics (SGLD) and the Consensus Algorithm \cite{Scott:2013}. SGLD and Consensus are seen to be strongly biased, whereas ZZ and ZZ-CV  target the correct distribution. However this comes at a cost: ZZ-CV loses much of its efficiency in this situation (due to the combination of lack of posterior contraction and unbounded Hessian); in particular it is not super-efficient. The use of multiple reference points may alleviate this problem, see also the discussion in Section~\ref{sec:discussion}.


\begin{figure}[ht!]
\begin{subfigure}[b]{0.48 \textwidth}
 \includegraphics[width=\textwidth]{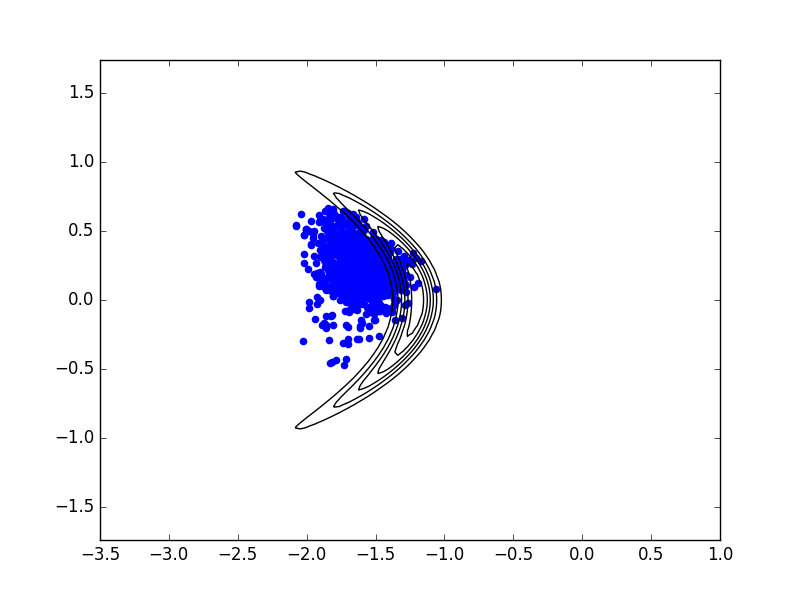}
 \caption{Consensus algorithm, 10 batches}
\end{subfigure}
\begin{subfigure}[b]{0.48 \textwidth}
 \includegraphics[width=\textwidth]{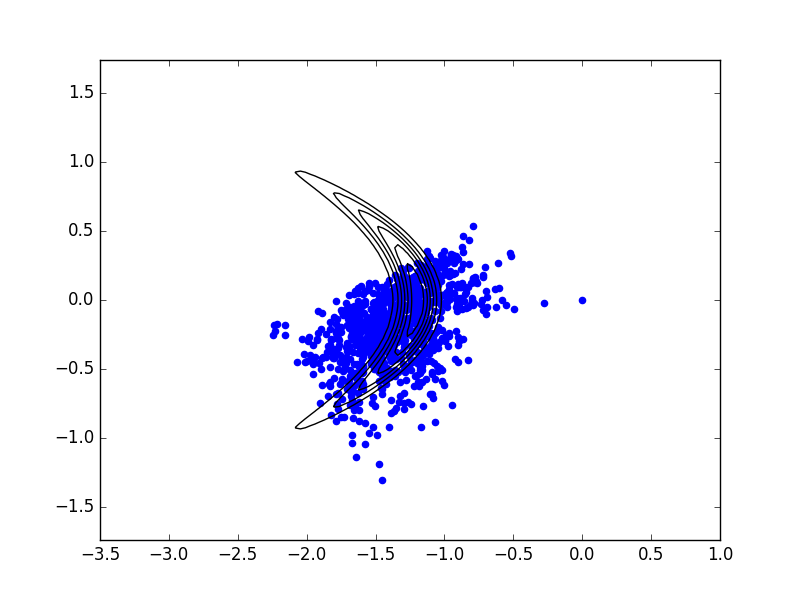}
 \caption{SGLD, 100 batches}
 \end{subfigure}\\
\begin{subfigure}[b]{0.48 \textwidth}
 \includegraphics[width=\textwidth]{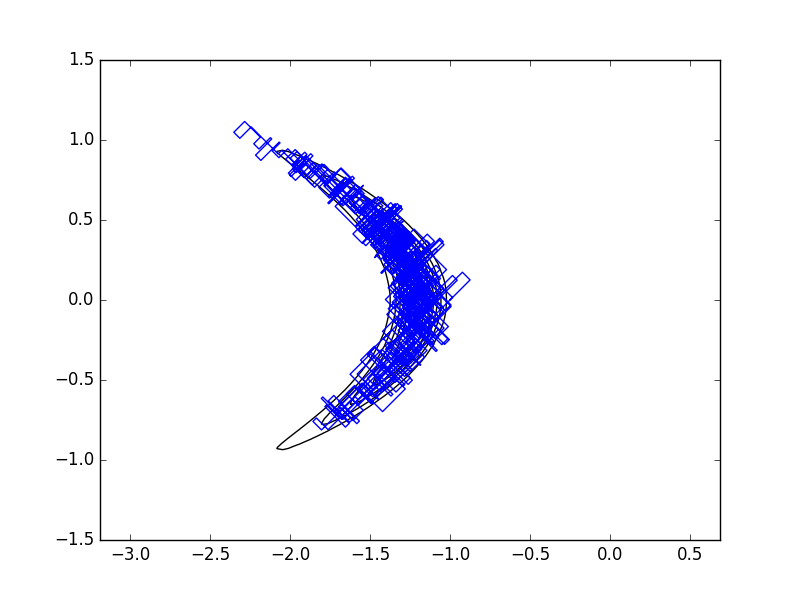}
 \caption{Zig-Zag}
\end{subfigure}
\begin{subfigure}[b]{0.48 \textwidth}
 \includegraphics[width=\textwidth]{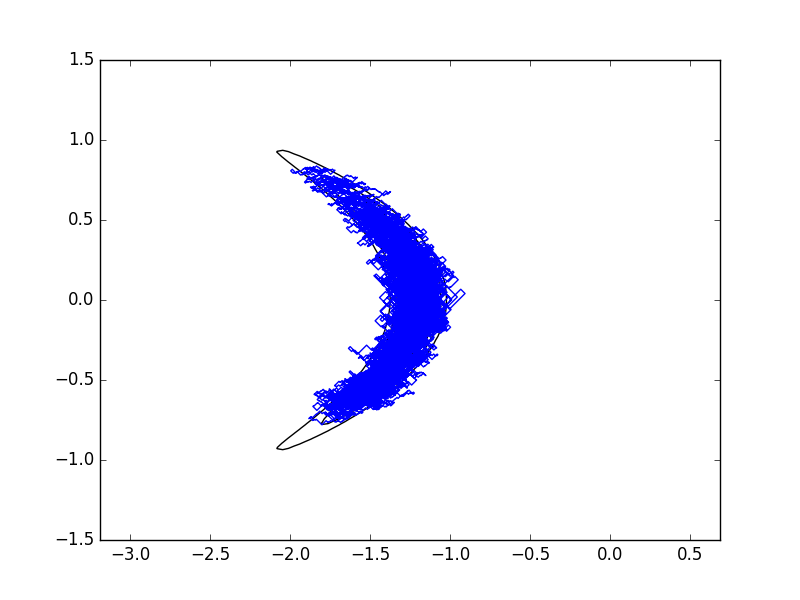}
 \caption{ZZ-CV}
\end{subfigure}
\caption{Trace plots of several algorithms (blue) and density contour plots for the non-identifiable logistic regression example of Section~\ref{sec:nonidentifiable}. In this example we have for the number of observations $n = 1,000$. Data is randomly generated from the model with true parameter satisfying $\xi_1 + \xi_2^2 = -1$. The prior is a 2-dimensional standard normal distribution. Due to the unbounded Hessian and because SGLD is not corrected by a Metropolis-Hastings accept/reject, the stepsize of SGLD needs to be set to a very small value (compared e.g. to what would be required for MALA) in order to prevent explosion of the trajectory; still the algorithm exhibits a significant asymptotic bias.}
\label{fig:nonidentifiable} 
\end{figure}

\section{Discussion}
\label{sec:discussion}

We have introduced the multi-dimensional Zig-Zag process and shown that it can be used as an alternative to standard MCMC algorithms. The advantages of the Zig-Zag process are that it 
is a non-reversible process, and thus has the potential to mix better than standard reversible MCMC algorithms, and that we can use sub-sampling ideas when simulating the process
and still be guaranteed to sample from the true target distribution of interest. We have shown that it is possible to implement sub-sampling with control-variates in a way that 
we can have super-efficient sampling from a posterior: the cost per effective sample size is sub-linear in the number of data points. We believe the latter aspect will be particularly
useful for applications where the computational cost of calculating the likelihood for a single data point is high.

As such, the Zig-Zag process holds substantial promise. However, being a completely new method, there are still substantial challenges in implementation which will need to be overcome for Zig-Zag to reach the levels of popularity of
standard discrete-time MCMC. The key challenges to implementing the Zig-Zag efficiently are
\begin{enumerate}
\item
to simulate from the relevant time-inhomogeneous Poisson process; and
\item
in order to realise the advantages of Zig-Zag for large datasets, reasonable centering points need to be found before commencing the MCMC algorithm itself.
\end{enumerate}
For the first of these challenges, we have shown how this can be achieved through bounding the rate of the Poisson
process, but the overall efficiency of the simulation algorithm then depends on how tight these bounds are.  In Subsection \ref{sec:compbd} we describe  efficient ways to carry this out. Moreover, as pointed out by a reviewer, there is a substantial literature on simulating 
stochastic processes that involve simulating such time-inhomogeneous Poisson processes \cite[]{gibson2000efficient,anderson2007modified}. Ideas from this literature could be leveraged both to extend the class of models for which we can simulate the Zig-Zag process, and also to make implementation of simulation algorithms more efficient. 

The second challenge applies when using the ZZ-CV algorithm to obtain super-efficiency for big data as discussed in Subsection \ref{sec:sub-sampling-control-variates}. Although in our experience finding appropriate centering points is rarely a serious problem, it is difficult to give a prescriptive recipe for this step.

On the face of it, these challenges may limit the practical applicability of Zig-Zag, at least in the short term. With that in mind, we have released an R/Rcpp package for logistic regression, as well as the code which reproduces the experiments of Section~\ref{sec:examples} \cite[]{zigzag-experiments}.

In addition, while Zig-Zag is an {\em exact approximate} simulation method, there are various short-cuts to speed it up at the expense of the introduction of an approximation. For instance, there are already ideas of approximately simulating the continuous-time dynamics, through approximate bounds on the Poisson rate \cite[]{pakman2016stochastic}. These ideas can lead to efficient simulation of the Zig-Zag process
for a wide class of models, albeit with the loss of exactness. Understanding the errors introduced by such an approach is an  open area. 

The most exciting aspect of the Zig-Zag process is the super-efficiency we observe when using sub-sampling with control variates. Already this idea has been adapted and shown to apply to other 
recent continuous-time MCMC algorithms \cite[]{Fearnhead:2018,pakman2016stochastic}. 
We have shown in Subsection \ref{sec:nonidentifiable} that Zig-Zag  can be applied effectively within highly non-Gaussian examples where rival approximate methods such as SGLD and the Consensus Algorithm are seriously biased. So there is no intrinsic reason to expect Zig-Zag to rely on the target distribution being close to Gaussian, although posterior contraction and the ability to find tight Poisson process rate bounds play important roles as we saw in our examples. There is much to learn about how the efficiency of  Zig-Zag depends on the statistical properties of the posterior distribution. However, unlike its approximate competitors, Zig-Zag will still remain an {\em exact approximate} method whatever the structure of the target distribution.

In truly `big data' settings, in principle we still need to process all the data once, although a suitable reference point can be determined using a subset of the data, we do need to evaluate the full gradient of the log density once at this reference point, and this computation is $O(n)$. This operation however is much easier to parallelize than MCMC is, and after this approximately independent samples can be obtained at a cost of $O(1)$ each. Thus if we wish to obtain $k$ approximately independent samples, the computational efficiency of ZZ-CV is $O(k + n)$ while the complexity of traditional MCMC algorithms is $O(k n)$. This is confirmed by the experiment in Section~\ref{sec:logistic}.


The idea for control variates we present in this paper is just one, possibly the simplest, implementation of this idea. There are natural extensions
to deal with e.g. multi-modal posteriors or situations where we do not have posterior concentration for all parameters. The simplest of these involve using multiple reference points and monitoring
the computational bound we get within the CV-ZZ algorithm and switching to a different algorithm when we stray so far from a reference point that this bound becomes too large. More sophisticated
approaches include using the ideas from \cite[]{dubey2016variance}, where we introduce a reference point for each data point and update the reference points for data within
the subsample at each iteration of the algorithm. This would lead to the estimate of the gradient that we center our control variate estimator around to depend on the recent history of the Zig-Zag process, and thus could be accurate even if we explore multiple modes or the tails of the target distribution. 

\subsection*{Acknowledgements}
The authors are grateful for helpful comments from referees, the editor and the associate editor which have improved the paper. Furthermore the authors acknowledge Matthew Moores (University of Warwick) for helpful advice on implementing the Zig-Zag algorithms  as an R package using Rcpp. All authors acknowledge the support of EPSRC under the {\em ilike} grant: EP/K014463/1.



\begin{supplement}[id=suppA]
  \sname{Supplement}
  \stitle{Supplement to ``The Zig-Zag Process and Super-Efficient
Sampling for Bayesian Analysis of Big Data''}
  \slink[doi]{COMPLETED BY THE TYPESETTER}
  \sdatatype{.pdf}
  \sdescription{Mathematics of the Zig-Zag process, scaling of SGLD, details on the experiments including how to obtain computational bounds.}
\end{supplement}

\bibliographystyle{imsart-nameyear}
\bibliography{./zigzag.bib}

\end{document}